\documentclass[11pt]{article}

\usepackage{amssymb,amsmath,xcolor,graphicx,xspace,colortbl,rotating} % ,revsymb4-1}
\usepackage{amsmath}  %% 010
\usepackage{amssymb}  %% 100
\usepackage{amssymb,amsmath,xcolor,graphicx,xspace,colortbl,rotating}  %% 100
\usepackage{boxedminipage}  %% 100
\usepackage{color,graphicx}  %% 100
\usepackage{comment}  %% 100
\usepackage{enumitem}  %% 100
\usepackage{fullpage}  %% 100
\usepackage{graphics}  %% 100
\usepackage{graphicx}  %% 100
\usepackage{latexsym}  %% 100
\usepackage{ragged2e}  %% 100
\usepackage{theorem}  %% 100
\usepackage{txfonts}  %% 100
\usepackage{xcolor}  %% 200
\usepackage{amsfonts}
\graphicspath{{KubaV3AA_graphics/}{KubaV3AA_tcache/}{KubaV3AA_gcache/}}
\DeclareGraphicsExtensions{.pdf,.eps,.ps,.png,.jpg,.jpeg}
\graphicspath{{KubaV3A_graphics/}{KubaV3A_tcache/}{KubaV3A_gcache/}}
\DeclareGraphicsExtensions{.pdf,.eps,.ps,.png,.jpg,.jpeg}
\graphicspath{{KubaV3A_graphics/}{KubaV3A_tcache/}{KubaV3A_gcache/}}
\DeclareGraphicsExtensions{.pdf,.eps,.ps,.png,.jpg,.jpeg}
\graphicspath{{KubaV3A_graphics/}{KubaV3A_tcache/}{KubaV3A_gcache/}}
\DeclareGraphicsExtensions{.pdf,.eps,.ps,.png,.jpg,.jpeg}
\newenvironment{proof}[1][Proof]
{\par\noindent{\bf #1:} }{\hspace*{\fill}\nolinebreak{$\Box$}\bigskip\par}

\newtheorem{theorem}{Theorem}
\theoremstyle{plain}

\newtheorem{conjecture}{Conjecture}
\newtheorem{corollary}{Corollary}

\newtheorem{lemma}{Lemma}

\numberwithin{equation}{section}
\begin{document}
\title{\textbf{On a Conjecture for a Hypergraph Edge Coloring Problem}}
\author{Wies\l aw Kubiak \protect\footnote{  Faculty of Business
Administration, Memorial University, St. John's, Canada. Email: wkubiak@mun.ca}
%Dominique de Werra\protect\footnote{EPFL, Lausanne, Switzerland. Email: dominique.dewerra@epfl.ch}
}
\maketitle
\begin{abstract}Let $H =(\mathcal{M} \cup \mathcal{J} ,E \cup \mathcal{E})$ be a hypergraph with two hypervertices $\mathcal{G}_1$ and $\mathcal{G}_2$ where $\mathcal{M} =\mathcal{G}_{1} \cup \mathcal{G}_{2}$ and $\mathcal{G}_{1} \cap \mathcal{G}_{2} =\varnothing $. An edge $\{h ,j\} \in E$ in a bi-partite multigraph graph $(\mathcal{M} \cup \mathcal{J} ,E)$ has an integer multiplicity $b_{j h}$, and a hyperedge $\{\mathcal{G}_{\ell } ,j\} \in \mathcal{E}$, $\ell=1,2$, has an integer multiplicity $a_{j \ell }$. It has been conjectured in \cite{KDWK} that $\chi  \prime  (H) =\lceil \chi  \prime _{f} (H)\rceil $, where $\chi  \prime  (H)$ and $\chi  \prime _{f} (H)$ are the edge chromatic number of $H$ and the fractional edge chromatic number of $H$ respectively. Motivation to study this hyperedge coloring conjecture comes from the University timetabling, and open shop scheduling with multiprocessors. We prove this conjecture in this paper.  
\end{abstract}

\section{The conjecture}
Let $H =(\mathcal{M} \cup \mathcal{J} ,E \cup \mathcal{E})$ be a hypergraph with two hypervertices $\mathcal{G}_1$ and $\mathcal{G}_2$ where $\mathcal{M} =\mathcal{G}_{1} \cup \mathcal{G}_{2}$ and $\mathcal{G}_{1} \cap \mathcal{G}_{2} =\emptyset $. An edge $\{h ,j\} \in E$ in a bi-partite multigraph graph $(\mathcal{M} \cup \mathcal{J} ,E)$ has an integer multiplicity $b_{j h}$, and a hyperedge $\{\mathcal{G}_{\ell } ,j\} \in \mathcal{E}$, $\ell=1,2$, has an integer multiplicity $a_{j \ell }$. 
We limit ourselves to the just defined hypergraphs $H$ in this paper.
%Let $H =(\mathcal{M} \cup \mathcal{J} ,E \cup \mathcal{E})$ be a hypergraph where $\mathcal{M} =\mathcal{G}_{1} \cup \mathcal{G}_{2}$ and $\mathcal{G}_{1} \cap \mathcal{G}_{2} =\varnothing $. An edge $\{h ,j\} \in E$ has an integer multiplicity $b_{j h}$ and a hyperedge $\{\mathcal{G}_{\ell } ,j\} \in \mathcal{E}$ has an integer multiplicity $a_{j \ell }$. 
It has been conjectured in \cite{KDWK} that $\chi  \prime  (H) =\lceil \chi  \prime _{f} (H)\rceil $, where $\chi  \prime  (H)$ and $\chi  \prime _{f} (H)$ are the edge chromatic number of $H$ and the fractional edge chromatic number of $H$ respectively, see \cite{SU} for more on the fractional graph theory.
%We prove this conjecture in this paper. 
Observe that $\mathcal{G}_1=\emptyset$ or $\mathcal{G}_2=\emptyset$ results in the $\chi  \prime  (H) = \chi  \prime _{f} (H)=\Delta(\mathcal{G}_2) + \chi' (\mathcal{M}\cup \mathcal{J}, E)$ and $\Delta(\mathcal{G}_1)+\chi' (\mathcal{M}\cup \mathcal{J}, E)$ respectively, where $\chi' (\mathcal{M}\cup \mathcal{J}, E)=\max\{\max_j\{\sum_h b_{jh}\},\max_h\{\sum_j b_{jh}\}\}$ is the edge chromatic number of the bi-partite multigraph $(\mathcal{M}\cup \mathcal{J}, E)$, and $ \Delta (\mathcal{G}_{\ell })$ = $\sum _{j \in \mathcal{J}}a_{j \ell }$ for $\ell  =1 ,2$. Thus the conjecture holds in this case and we assume non-empty $\mathcal{G}_1$ and non-empty $\mathcal{G}_2$ from now on in the paper.

A feasible edge coloring in $H$ can be partitioned in the following four parts: part (a) includes matchings with hyperedges $(\mathcal{G}_{1} ,j)$ for some $j \in \mathcal{J}$, and  edges $(h ,j)$ where $h \in \mathcal{G}_{2}$ and $j \in \mathcal{J}$; part (b) includes matchings with hyperedges $(\mathcal{G}_{1} ,j)$, and hyperedges $(\mathcal{G}_{2} ,j \prime )$ for some $j ,j \prime  \in \mathcal{J}$; part (c) includes matchings with hyperedges $(\mathcal{G}_{2} ,j)$ for some $j \in \mathcal{J}$, and edges $(h ,j)$ where $h \in \mathcal{G}_{1}$ and $j \in \mathcal{J}$; and part (d) includes matchings with edges $(h ,j)$ only where $h \in \mathcal{M}$ and $j \in \mathcal{J}$. The parts (a), (b), (c) and (d) have multiplicities $ \Delta (\mathcal{G}_{1}) -r$, $r$, $ \Delta (\mathcal{G}_{2}) -r$, and $w$ respectively, for some $r$ and $w$. Therefore the total of $ \Delta (\mathcal{G}_{1}) + \Delta (\mathcal{G}_{2}) -r +w$ colors are used, and the minimization of the number of colors required to color the edges of $H$ reduces to the minimization of $w -r$. Following the convention used in \cite{DWKK}
and \cite{KDWK} we refer to $h \in \mathcal{M}$ as machine $h$, and to $\jmath \in \mathcal{J}$ as job $j$ for convenience. It was shown in \cite{DWKK} that $\chi \prime (H)> \chi \prime_f (H)$ for some hypergraphs $H$, and in \cite{KDWK} that $\lceil \chi \prime_f(H)\rceil +1 \geq \chi \prime (H)$ for each hypergraph $H$. 

Let $m =\vert \mathcal{M}\vert $ be the number of machines and $n =\vert \mathcal{J}\vert $ be the number of jobs. Without loss of generality $\mathcal{J}=\{1,...,n\}$ and $\mathcal{M}=\{1,...,m\}$.
The following integer linear program $\cal{ILP}$ with variables $r$, $w$, and $y_{j h}$, $x_{j \ell }$, for $j \in \mathcal{J}$, $h \in \mathcal{M}$, and $\ell  =1 ,2$, and integer coefficients $b_{j h}$, $a_{j \ell }$, for $j \in \mathcal{J}$, $h \in \mathcal{M}$,  and $\ell  =1 ,2$, 
%and $ \Delta (\mathcal{G}_{\ell })$ = $\sum _{j =1}^{n}a_{j \ell }\;\text{,}\;$ 
for $\ell  =1 ,2$ was given in \cite{DWKK} and \cite{KDWK} to find $\chi \prime (H)$.

\begin{equation}ILP =\min (w -r) \label{obj}
\end{equation}Subject to
\begin{equation}\sum _{j}b_{j h} -( \Delta (\mathcal{G}_{2}) -r) \leq \sum _{j}y_{j h} \leq w\;\;\text{\ \ }\;\; \;\;\text{\ }\;\;h \in \mathcal{G}_{1} \label{B2}
\end{equation}
\begin{equation}\sum _{j}b_{j h} -( \Delta (\mathcal{G}_{1}) -r) \leq \sum _{j}y_{j h} \leq w\;\;\text{\ \ \ }\;\;\;\;\text{}\;\;h \in \mathcal{G}_{2} \label{B3}
\end{equation}
\begin{equation}\sum _{h}y_{j h} \leq w\;\;\text{\ }\;\;\;\;\text{\ }\;\;j\in \mathcal{J} \label{B4}
\end{equation}
\begin{equation}0 \leq y_{j h} \leq b_{j h} \;\;\text{\ }\;\; \;\;\text{}\;\;h \in \mathcal{M}\;\; \;\;\text{\ }\;\;j \in \mathcal{J}\label{B5}
\end{equation}
\begin{equation}\sum _{j}x_{j 1} =r \label{B6}
\end{equation}
\begin{equation}\sum _{j}x_{j 2} =r \label{B7}
\end{equation}
\begin{equation}x_{j 1} +x_{j 2} \leq r\;\;\text{\ \ }\;\;\;\;\text{\ }\;\;j \in \mathcal{J}\label{B8}
\end{equation}
\begin{equation}0 \leq x_{j \ell } \leq a_{j \ell }\;\; \;\;\text{\ }\;\;j \in \mathcal{J} \;\; \;\;\text{\ }\;\;\ell=1,2 \label{B9}
\end{equation}
\begin{equation}\sum _{h \in \mathcal{G}_{1}}(b_{j h} -y_{j h}) +a_{j 2} -x_{j 2} \leq  \Delta (\mathcal{G}_{2}) -r\;\;\text{\ \ \ }\;\; \;\;\text{\ }\;\;j \in \mathcal{J}\label{B10}
\end{equation}
\begin{equation}\sum _{h \in \mathcal{G}_{2}}(b_{j h} -y_{j h}) +a_{j 1} -x_{j 1} \leq  \Delta (\mathcal{G}_{1}) -r\;\;\text{\ \ \ }\;\; \;\;\text{\ }\;\;j \in \mathcal{J}\label{B11}
\end{equation}

The variable $y_{jh}$  represents the amount of $j\in \mathcal{J}$ on $h\in \mathcal{M}$ in part (d). The variable $x_{j\ell}$ represents the amount of $j\in \mathcal{J}$ on $\mathcal{G}_{\ell}$, $\ell=1,2$, in part (b). The variable $w$ is the size of (d), and the variable $r$ is the size of (b). The constraints (\ref{B2})-(\ref{B5}) guarantee that the size of part (d) does not exceed $w$. The constraints (\ref{B6})-(\ref{B9}) guarantee that the size of part (b) equals $r$. The constraints (\ref{B10})-(\ref{B11}) along with the left hand side inequalities in (\ref{B2}) and (\ref{B3}) guarantee that the size of part (a) does not exceed $\Delta(\mathcal{G}_1) - r$ and that the size of part (c) does not exceed $\Delta(\mathcal{G}_2) - r$.

Let $ILP$ be the value of optimal solution to this program, and let $L P$ be the value of optimal solution to the $L P$-relaxation of this program. The conjecture $\chi  \prime  (H) =\lceil \chi  \prime _{f} (H)\rceil $ for each hypergraph $H$
%, to find the edge chromatic number $\chi  \prime  (H)$ of $H$ 
is therefore equivalent to the following conjecture that we show in this paper to hold.

\begin{conjecture}\label{conj}
$I L P =\lceil L P\rceil $. 
\end{conjecture}

Motivation to study this hyperedge coloring problem comes from the University timetabling studied
in \cite{ADW} and \cite{DWAD}, and open shop scheduling with multiprocessors studied in \cite{DWKK}, \cite{KDWK},
see also  \cite{LLV}. The University timetabling is a generalization of the well-known class-teacher timetabling model. In the generalization an edge represent a single-period lecture given by a teacher to
a class, and a hyperedge represent a single-period lecture given by a teacher to a group of classes
simultaneously. One looks for a minimum number of periods in which to complete all lecture without conflicts.
In the open shop scheduling with multiprocessors 
%studied in \cite{DWKK}, \cite{KDWK},
%see also  \cite{LLV}, where 
a set of jobs $\mathcal{J}=\{J_1,...,J_n\}$ is scheduled on machines $\mathcal{M}=\{M_1,...,M_m\}$. The set of machines is partitioned into two groups $\mathcal{G}_1$ and $\mathcal{G}_2$. Each job consists of \emph{single-processor} and \emph{multiprocessor} operations. A single-processor operation requires one of the machines in $\mathcal{M}$, and a multiprocessor operation requires \emph{all} machines from the group, $\mathcal{G}_1$ or $\mathcal{G}_2$.  Each machine can execute at most one operation at a time, and no two operations of the same job can be executed simultaneously. Any operation can be preempted at any moment though we limit preemptions to integer points for the edge coloring, and at any point for the fractional edge coloring. The makespan is to be minimized. Please see \cite{ADW}, \cite{DWAD}, \cite{DWKK}, \cite{KDWK}, and  \cite{LLV} for more on the applications of the hyperedge coloring problem.

It was pointed out in \cite{DWKK} that the  hypergraphs considered in this paper generalize bipartite multigraphs, but they do not belong to known classes, like balanced, normal or with the K\H{o}nig-Egerv\'{a}ry property, see \cite{B}.

\section{Fractional edge colorings with $\lceil LP \rceil$ colors}
%The Optimal Solutions with Integral $r$\bigskip 
%}
Let $(\mathbf{y}^{ \ast } ,\mathbf{x}^{ \ast }, w^{ \ast } ,r^{ \ast })$ be an optimal solution to the $L P$-relaxation of $\cal{ILP}$. Let $w^*=\lfloor w^* \rfloor + \lambda_{w^*}$ and $r^*=\lfloor r^* \rfloor + \lambda_{r^*}$, where $0\leq \lambda_{w^*}<1$
and $0\leq \lambda_{r^*}<1$. Consider the following linear program $\ell p$.

\begin{equation*}l p =\min r
\end{equation*}Subject to
\begin{equation}w -r =\left \lceil w^{ \ast } -r^{ \ast }\right \rceil  \label{B20}
\end{equation}
\begin{equation}
\lfloor r^* \rfloor \leq r
\end{equation}
\begin{equation}\sum _{j}b_{j h} -( \Delta (\mathcal{G}_{2}) -r) \leq \sum _{j}y_{j h} \leq w\;\;\text{\ \ \ }\;\; \;\;\text{}\;\;h \in \mathcal{G}_{1} \label{y1}
\end{equation}
\begin{equation}\sum _{j}b_{j h} -( \Delta (\mathcal{G}_{1}) -r) \leq \sum _{j}y_{j h} \leq w\;\;\text{\ \ }\;\; \;\;\text{\ }\;\;h \in \mathcal{G}_{2} \label{y2}
\end{equation}
\begin{equation}\sum _{h}y_{j h} \leq w\;\;\text{\ }\;\; \;\;\text{\ }\;\;j \in \mathcal{J}\label{r1d}
\end{equation}
\begin{equation}0 \leq y_{j h} \leq b_{j h} \;\;\text{\ }\;\; \;\;\text{}\;\;h \in \mathcal{M}\;\; \;\;\text{\ }\;\;j \in \mathcal{J}\label{y4}
\end{equation}
\begin{equation}\sum _{j}x_{j 1} =r \label{bm1}
\end{equation}
\begin{equation}\sum _{j}x_{j 2} =r \label{bm2}
\end{equation}
\begin{equation}x_{j 1} +x_{j 2} \leq r\;\;\text{\ }\;\; \;\;\text{\ \ }\;\;j \in \mathcal{J}\label{y3}
\end{equation}
\begin{equation}0 \leq x_{j \ell } \leq a_{j \ell } \;\; \;\;\text{\ }\;\;j \in \mathcal{J} \;\; \;\;\text{\ }\;\;\ell=1,2\label{y5}
\end{equation}
\begin{equation}\sum _{h \in \mathcal{G}_{1}}(b_{j h} -y_{j h}) +a_{j 2} -x_{j 2} \leq  \Delta (\mathcal{G}_{2}) -r\;\;\text{\ \ }\;\; \;\;\text{\ \ }\;\;j\in \mathcal{J} \label{r1a}
\end{equation}
\begin{equation}\sum _{h \in \mathcal{G}_{2}}(b_{j h} -y_{j h}) +a_{j 1} -x_{j 1} \leq  \Delta (\mathcal{G}_{1}) -r\;\;\text{\ \ }\;\; \;\;\text{\ \ }\;\;j\in \mathcal{J} \label{r1c}
\end{equation}

Let $(\mathbf{y} ,\mathbf{x} ,r ,w)$ be an optimal solution to $\ell p$.  The solution exists since $(\mathbf{y}^* ,\mathbf{x} ^*,r^* ,\lfloor w^* \rfloor +\lambda_{r^*})$ is feasible for $\ell p$ if $\lambda_{w^*}\leq\lambda_{r^*}$, and
$(\mathbf{y}^* ,\mathbf{x} ^*,r^* ,\lceil w^* \rceil +\lambda_{r^*})$ is feasible for $\ell p$ if $\lambda_{w^*}>\lambda_{r^*}$, thus $\ell p$ is feasible and clearly it is also bounded. Observe that $w^*\leq \lfloor w^* \rfloor+\lambda_{r^*}$ for $\lambda_{w^*}\leq\lambda_{r^*}$, and $\lceil w^* \rceil +\lambda_{r^*} - r^*=\lceil w^* \rceil  - \lfloor r^* \rfloor=\lceil w^* - r^* \rceil$
for $\lambda_{w^*} >\lambda_{r^*}$. Thus the $\ell p$ gives a fractional edge coloring of the hypergraph with $\lceil LP \rceil$ colors.

We assume without loss of generality that the solution meets the machine \emph{saturation} condition, i.e. the upper and lower bounds in (\ref{y1})
and (\ref{y2}) are equal. If the machine saturation is not met by the solution for some machine $h$, then a job $j (h)$ with $b_{j (h) ,h} =w -\sum _{j}b_{j h} +( \Delta (\mathcal{G}_{2}) -r)$, $a_{j(h)1}=a_{j(h)2}=0$ should be added to the instance for each such machine to make the solution meet the saturation condition. Observe that by (\ref{B20})
$b_{j (h) ,h}$ is integral so the extended instance is a valid instance of the hypergraph edge coloring problem. We take $y_{j (h) ,h} =w -\sum _{j}y_{j h}$ in the extended solution. Observe that $n=|\mathcal{J}|\geq |\mathcal{G}_1|+|\mathcal{G}_2|$ for the solutions that meet the saturation condition.  

An optimal solution $(\mathbf{y} ,\mathbf{x} ,r ,w)$ to $\ell p$ that is integral is feasible for $\cal{ILP}$, and
$w - r =\lceil w^{ \ast } -r^{ \ast }\rceil  =\lceil L P\rceil $.  Moreover this solution is optimal for $\cal{ILP}$ since by definition of $L P$-relaxation we have $L P \leq I L P$ for any feasible solution to $\cal{ILP}$. This proves Conjecture \ref{conj}. Therefore it suffices to prove that there is an optimal solution to $\ell p$ that is integral. To that end, we prove the following theorem in the remainder of the paper.
%We have the following lemma. 

%\begin{lemma}
%\label{L1} If the $r$ in an optimal solution to the $l p$ is integral, then the conjecture holds, and the $I L P$ can be solved in polynomial time. 
%\end{lemma}

%\begin{proof}
%By (\ref{B20}), if $r$ is integer, then $w$ is integer. Using the projective method of \cite{KDWK}
%we can then find in polynomial time an integral $\mathbf{y}$, and finally given integral
%$r ,w$ and $\mathbf{y}$ we can find in polynomial time an integral $\mathbf{x}$. Therefore, we can find in polynomial time an integer solution to the $l p$ provided that $r$ is integer. Since $\lceil w^{ \ast } -r^{ \ast }\rceil  =\lceil L P\rceil $ and $(\mathbf{y} ,\mathbf{x} ,r ,w)$ is a feasible solution to $I L P$ such that $I L P =\lceil L P\rceil $, the conjecture holds. Finally by definition of $L P$-relaxation we have $L P \leq I L P$ for any feasible solution to $I L P$. Thus $(\mathbf{y} ,\mathbf{x} ,r ,w)$ is an optimal solution to $I L P$ obtained in polynomial time. This proves the lemma. 
%\end{proof}

%By Lemma \ref{L1}
%it suffices to prove that $r$ is integral in an optimal solution to $l p$. 

\begin{theorem}
\label{integer} The $r$ in an optimal solution to $\ell p$ is integral.  
Moreover, there is optimal solution to $\ell p$ that is integral.
\end{theorem}

\begin{proof}
Let $\mathbf{s}=(\mathbf{y} ,\mathbf{x} ,r ,w)$ be an optimal solution to $\ell p$. Suppose for a contradiction that the $r$ in $\mathbf{s}$ equals
\begin{equation*}r =\left \lfloor r^*\right \rfloor  +\epsilon 
\end{equation*}where $0 <\epsilon  <1.$ Thus by (\ref{B20})
\begin{equation*}w =\left \lfloor w\right \rfloor  +\epsilon \;\text{.}\;
\end{equation*}In the remaining sections
of the paper we show that such $\mathbf{s}$ can not be optimal which leads to a contradiction and proves the first part of the theorem. We then show that an optimal solution that is integral can be  found in polynomial time. An outline of the proof will be given at the end of the next section after we first introduce the necessary notations and definitions.
\end{proof}

%hus there is a feasible integral solution $(\mathbf{y} ,\mathbf{x} ,r ,w)$ to $lp$ with $r=\lfloor r^ * \rfloor<r$. Hence there is a feasible solution
%to $lp$ which is smaller than optimal $r$  which gives contradiction. 

% obtained in polynomial time. 

\section{Preliminaries}
%Let $(\mathbf{y} ,\mathbf{x} ,r ,w)$ be an optimal solution to $l p$ with $0 <\epsilon  <1.$ 
Consider the solution $\mathbf{s}=(\mathbf{y} ,\mathbf{x} ,r ,w)$. Let $B_{1}$ be the set of all jobs $j$ with fractional $x_{j 1}$, 
%i.e. $x_{j1}=\lfloor x_{j1} \rfloor + \varepsilon_{j}$,  $0<\varepsilon_j<1$,
and let $B_{2}$ be the set of all jobs $j$ with fractional $x_{j 2}$.
%, i.e. $x_{j1}=\lfloor x_{j2} \rfloor + \varepsilon_{j}$,  $0<\varepsilon_j<1$
Clearly both sets are non-empty
for $\epsilon  >0$. By (\ref{bm1}) and (\ref{bm2}) the fractions in $B_{\ell}$ sum up to $i_{\ell}+ \epsilon$, where $i_{\ell}$ is non-negative integer, for $\ell=1,2$.

\noindent A job $j$ is \emph{d-tight} if
\begin{equation*}\sum _{h}y_{j h} =w.
\end{equation*}Denote by $D$ the set of all $d$-tight jobs. 

\noindent A job $j$ is \emph{a-tight} if
\begin{equation*}\sum _{h \in \mathcal{G}_{2}}(b_{j h} -y_{j h}) +a_{j 1} -x_{j 1} = \Delta (\mathcal{G}_{1}) -r.
\end{equation*}

\noindent A
job $j$ is \emph{c-tight} if
\begin{equation*}\sum _{h \in \mathcal{G}_{1}}(b_{j h} -y_{j h}) +a_{j 2} -x_{j 2} = \Delta (\mathcal{G}_{2}) -r.
\end{equation*}

%\noindent A
%job $j$ is \emph{b-tight} if
%\begin{equation}x_{j 1} +x_{j 2} =r.
%\end{equation}

\noindent For
jobs $g$ and $k$ such that $x_{g 1} >0$ and $x_{k 2} >0$ define
\begin{equation*}\varepsilon _{r} (g ,k) =\left \{\begin{array}{cc}\min_{j \in (B_{1} \cup B_{2}) \setminus \{g ,k\}}\{r -(x_{j 1} +x_{j 2}) ,\epsilon \} & \;\text{if}\;(B_{1} \cup B_{2}) \setminus \{g ,k\} \neq \emptyset  ; \\
\epsilon  & \;\text{if}\;B_{1} \cup B_{2} \subseteq \{g ,k\}\;\text{.}\;\end{array}\right .
\end{equation*}Observe that jobs $g$ and $k$ with $\varepsilon _{r} (g ,k) >0$ can potentially be used to obtain a solution to $l p$ with smaller $r$ since the reduction of both $x_{g 1}$ and $x_{k 2}$ by some small enough $\varepsilon  >0$ will leave the resulting constraint (\ref{y3}) satisfied. Moreover define
\begin{equation*}\varepsilon _{c} (k) =\sum _{h \in \mathcal{G}_{1}}y_{k h} -(\sum _{h \in \mathcal{G}_{1}}b_{k h} +a_{k 2} -x_{k 2} - \Delta (\mathcal{G}_{2}) +r),
\end{equation*}
\begin{equation*}\varepsilon _{a} (g) =\sum _{h \in \mathcal{G}_{2}}y_{g h} -(\sum _{h \in \mathcal{G}_{2}}b_{g h} +a_{g 1} -x_{g 1} - \Delta (\mathcal{G}_{1}) +r).
\end{equation*}

Let
$G$ be a job-machine bi-partite graph such that there is an edge between machine $h \in \mathcal{M}$ and job $j \in \mathcal{J}$ if and only if $y_{j h} >0$. A \emph{column} $I =(M_{I} ,\epsilon _{I})$ consists of a matching $M_{I}$ in $G$ that matches all $m$ machines in $\mathcal{M}$ with a subset of exactly $m$ jobs in $\mathcal{J}$, and its multiplicity $\epsilon _{I} >0$. Let $\mathcal{J}_{I}$ be the set of all jobs matched in $M_{I}$, i.e. $\mathcal{J}_{I} =\{j \in \mathcal{J} :(j ,h) \in M_{I}\;\;\text{for some}\;\;h \in \mathcal{M}\}$. By definition of $D$ we require that $D \subseteq \mathcal{J}_{I}$ for a column in $\mathbf{s}$. By \cite{GS}, see also \cite{GK},
part (d) can be represented by a set of columns $d (\mathbf{y} ,w) =\{I_{1} ,\ldots  ,I_{p}\}$. For a set $X$ of columns let $l (X)$ denote the total multiplicity of all columns in $X$. We have $l (d (\mathbf{y} ,w)) =w$. Let $I_1=(M_{I_1},\varepsilon_{I_1}),...,I_q=(M_{I_q},\varepsilon_{I_q})$ be a subset of $q\geq 1$ columns from $d (\mathbf{y} ,w)$, the set of columns  $(M_{I_1},\lambda_1),...,(M_{I_q},\lambda_q)$, where  $0\leq\lambda_1\leq \varepsilon_{I_1},...,0\leq \lambda_q\leq \varepsilon_{I_q}$ and $\lambda_1+...+\lambda_q=\lambda$ is called the \emph{interval} of length $\lambda$ in $d (\mathbf{y} ,w)$.

Let $u_{1} ,\ldots  ,u_{p}$ and $l_{1} ,\ldots  ,l_{q}$ be different jobs from $\mathcal{J}$, and $I$ be a column. We say that $I$ is of type
\begin{equation*}\binom{ \ast  ,u_{1} ,\ldots  ,u_{p}}{ \ast  ,l_{1} ,\ldots  ,l_{q}}
\end{equation*}if $\{(u_{1} ,h_{1}) ,\ldots  ,(u_{p} ,h_{p})\} \subseteq M_{I}$ for some machines $h_{1} ,\ldots  ,h_{p}$ in $\mathcal{G}_{1}$, and $\{(l_{1} ,H_{1})$ $ ,\ldots  ,(l_{q} ,H_{q})\} \subseteq M_{I}$ for some machines $H_{1} ,\ldots  ,H_{q}$ in $\mathcal{G}_{2}$. For convenience, we sometimes use the following notation
\begin{equation*}\binom{ \ast  ,U}{ \ast  ,L}
\end{equation*}
where $U=\{u_{1} ,\ldots  ,u_{p}\}$ and $L=\{l_{1} ,\ldots  ,l_{q}\}$.
By definition if $p =0$ or $q =0$, then the asterisk alone denotes any matching on $\mathcal{G}_{1}$ or $\mathcal{G}_{2}$ respectively. We extend this notation for convenience as follows. Let $u$ and $l$ be different jobs from $\mathcal{J}$, and $I$ be a column. We say that $I$ is of type
\begin{equation*}\binom{ \ast  ,\overline{u}}{ \ast  ,\overline{l}}
\end{equation*}if $(u ,h) \notin M_{I}$ for any machine $h \in \mathcal{G}_{1}$, and $(l ,H) \notin M_{I}$ for any machine $H \in \mathcal{G}_{2}$. 

The outline of the proof is as follows. Sections \ref{S4} and \ref{S5} characterize those columns that cannot occur in $\mathbf{s}$ with $\epsilon>0$ since their presence would contradict the optimality of $r$. Section \ref{S6} shows that each job in $B_1$ must be both $a$-tight and $d$-tight,
and each job in $B_2$ must be both $c$-tight and $d$-tight in $\mathbf{s}$ using the characterization. Section  \ref{crossi} proves that $B_1\cap B_2=\emptyset$  in $\mathbf{s}$. Section \ref{S8} shows that  the product $x_{j1}x_{j2}=0$  in $\mathbf{s}$ except for the case where $B_1=\{j\}$ or $B_2=\{j\}$. Section \ref{S10} proves that the fractions in $B_1$ or $B_2$ may not sum up to $\epsilon$, ($i_1=0$ or $i_2=0$), since if they did the $\mathbf{s}$ would not be optimal. The proof relies on the results of earlier sections. Section \ref{S12}   proves that the fractions in $B_1$ and $B_2$ may not sum up to $i_1+\epsilon$ and $i_2+\epsilon$, ($i_1>0$ and $i_2>0$), since again if they did the $\mathbf{s}$ would not be optimal. The proof relies on the projections introduced in Section \ref{S11} and the results of earlier sections. Section \ref{S13} summarizes these results proving that the $r$ in $\mathbf{s}$ must be integral, $\epsilon=0$. The section also shows how to obtain an integral optimal solution for $\ell p$ using the network flow problems from Sections \ref{S10} and \ref{S12}, and the integrality of $r$ and $w$.

\section{Columns absent from $d (\mathbf{y} ,w)$ in $\mathbf{s}$} \label{S4}
In this section we show that for two different jobs $g$ and $k$ jobs such that $x_{g 1} >0$ and $x_{k 2} >0$ certain columns or subsets of columns must be missing from $d (\mathbf{y} ,w)$ if $\epsilon>0$. Though these results are contingent on $\varepsilon _{r} (g ,k) >0$, we show that this condition often holds, for instance in Section \ref{crossi} we show that this inequality holds for each pair $g \in B_{1}$ and $k \in B_{2}$. 

Let $g$ and $k$ be two different jobs such that $x_{g 1} >0$ and $x_{k 2} >0$. A \emph{$(g ,k)$-feasible} \emph{semi-matching} in $G$ is a set of edges $E$ of $G$ of cardinality $m=|\mathcal{G}_1|+|\mathcal{G}_2|$ such that
\begin{enumerate}
\item $E_{1} =\{(j ,h) \in E :h \in \mathcal{G}_{1}\}$ and $E_{2} =\{(j ,h) \in E :h \in \mathcal{G}_{2}\}$ are matchings; 

\item there are $h \in \mathcal{M}$ and $(j ,h) \in E$ for each $j \in D$; 

\item if $\varepsilon _{a} (g) =0$, then $(g ,h) \notin E_{2}$ for any $h \in \mathcal{G}_{2}$; 

\item if $\varepsilon _{c} (k) =0$, then $(k ,h) \notin E_{1}$ for any $h \in \mathcal{G}_{1}$. \end{enumerate}

If $E$ is a matching, then a $(g ,k)$-feasible semi-matching in $G$ is called a \emph{$(g ,k)$-feasible} \emph{matching} in $G$.

We define solution $(\mathbf{y} (E) ,\mathbf{x} (g ,k) ,r (g ,k) ,w (g ,k) ,\varepsilon )$ for jobs $g$, $k$, and a $(g ,k)$-feasible semi-matching $E$, where 

\begin{equation}\varepsilon  =\left \{\begin{array}{cc}\varepsilon'  & \;\text{if}\;\varepsilon _{a} (g) =0\;\text{and}\;\varepsilon _{c} (k) =0\;\text{;}\; \\
\min \{\varepsilon' ,\varepsilon _{a} (g)\} & \;\text{if}\;\varepsilon _{a} (g) >0\;\text{and}\;\varepsilon _{c} (k) =0\;\text{;}\; \\
\min \{\varepsilon',\varepsilon _{c} (k)\} & \;\text{if}\;\varepsilon _{a} (g) =0\;\text{and}\;\varepsilon _{c} (k) >0\;\text{;}\; \\
\min \{\varepsilon',\varepsilon _{a} (g) ,\varepsilon _{c} (k)\} & \;\text{if}\;\varepsilon _{a} (g) >0\;\text{and}\;\varepsilon _{c} (k) >0\;\text{;}\;\end{array}\right . \label{epsy}
\end{equation}and
\begin{equation}\varepsilon'  =\min \{\epsilon  ,\varepsilon _{r} (g ,k) ,x_{g 1} ,x_{k 2} ,\min_{(j ,h) \in E}\{y_{j h}\} ,\min_{j \in \mathcal{J} \setminus D}\{r -\sum _{h}y_{j h}\}\}\;\text{,}\; \label{x1}
\end{equation}as follows
\begin{equation}y_{j h} (E) =\left \{\begin{array}{cc}y_{j h} -\varepsilon  & \;\text{if}\;(j ,h) \in E\;\text{;}\; \\
y_{j h} & \;\text{otherwise}\; ;\end{array}\right . \label{x2}
\end{equation}
\begin{equation}x_{j 1} (g ,k) =\left \{\begin{array}{cc}x_{g 1} -\varepsilon  & \;\text{if}\;j =g\;\text{;}\; \\
x_{j 1} & \;\text{if}\;j \neq g ;\end{array}\right . \label{x3}
\end{equation}
\begin{equation}x_{j 2} (g ,k) =\left \{\begin{array}{cc}x_{k 2} -\varepsilon  & \;\text{if}\;j =k\;\text{;}\; \\
x_{k 2} & \;\text{if}\;j \neq k ;\end{array}\right . \label{x4}
\end{equation}

\begin{equation}r (g ,k) =r -\varepsilon  ; \label{x5}
\end{equation}

\begin{equation}w (g ,k) =w -\varepsilon \;\text{.}\; \label{x6}
\end{equation}

We have the following lemma

\begin{lemma}
\label{reduc} Let $g$ and $k$ be two different jobs such that $x_{g 1} >0$ and $x_{k 2} >0$. If $\varepsilon _{r} (g ,k) >0$, then no $(g ,k)$-feasible semi-matching $E$ in $G$ exists. 
\end{lemma}

\begin{proof}
Suppose for a contradiction that a $(g ,k)$-feasible semi-matching $E$ in $G$ exists. It suffices to show that solution $s =(\mathbf{y} (E) ,\mathbf{x} (g ,k) ,r (g ,k) ,w (g ,k) ,\varepsilon )$ is feasible for $l p$. The feasibility of $s$ implies that $r$ is not optimal for $l p$ since by the lemma assumptions $\epsilon  >0$ and thus $r (g ,k) <r$ which gives a contradiction. To prove the feasibility of $s$ we observe that since both $E_{1}$ and $E_{2}$ are matchings and $\vert E_{1}\vert  +\vert E_{2}\vert  =m$ we have
\begin{equation*}\sum _{j}y_{j h} (E) =\sum _{j}y_{j h} -\varepsilon \;\;\text{\ \ }\;\; \;\;\text{}\;\;h \in \mathcal{G}_{1}
\end{equation*}and
\begin{equation*}\sum _{j}y_{j h} (E) =\sum _{j}y_{j h} -\varepsilon \;\;\text{\ \ }\;\; \;\;\text{}\;\;h \in \mathcal{G}_{2}\;\text{.}\;
\end{equation*}Thus both (\ref{y1})
and (\ref{y2}) hold for $s$ by (\ref{x5}) and (\ref{x6}). The (\ref{x5})
and (\ref{x6}) also imply (\ref{B20}) for $s$. The (\ref{x3}) and (\ref{x4}) imply (\ref{bm1})
and (\ref{bm2}) for $s$. Since $\varepsilon _{r} (g ,k) >0$, definition of $\varepsilon $ guarantees (\ref{y3}) for $s$. The $E$ covers each job in $D$ at least once, thus (\ref{r1d}) holds for each job in $D$ in $s$, and by definition of $\varepsilon $ for each job in $\mathcal{J} \setminus D$. Clearly, the (\ref{y4}) and (\ref{y5}) are met for $s$ by the lemma assumptions and definition of $\varepsilon $. Therefore it remains to show that (\ref{r1a}) and (\ref{r1c})
hold for $s$. First, since both $E_{1}$ and $E_{2}$ are matchings we have 

\begin{equation*}\sum _{h \in \mathcal{G}_{1}}y_{j h} -\varepsilon  \leq \sum _{h \in \mathcal{G}_{1}}y_{j h} (E) \leq \sum _{h \in \mathcal{G}_{1}}y_{j h}\;\;\text{\ \ }\;\; \;\;\text{}\;\;j\in \mathcal{J}
\end{equation*}and
\begin{equation*}\sum _{h \in \mathcal{G}_{2}}y_{j h} -\varepsilon  \leq \sum _{h \in \mathcal{G}_{2}}y_{j h} (E) \leq \sum _{h \in \mathcal{G}_{2}}y_{j h}\;\;\text{\ \ }\;\; \;\;\text{}\;\;j\in \mathcal{J}\;\text{.}\;
\end{equation*}Hence (\ref{r1a})
and (\ref{r1c}) hold for all jobs in $\mathcal{J} \setminus \{g ,k\}$ for $s$, and for job $g$ provided that the slack $\varepsilon _{a} (g)$ of $g$ in $(a)$ is positive, and for $k$ provided that the slack $\varepsilon _{c} (k)$ of $k$ in $(c)$ is positive. The (\ref{r1a}) and (\ref{r1c})
hold for $g$ and $k$ with no slack in $(a)$ and $(c)$ respectively as well since then
\begin{equation*}\sum _{h \in \mathcal{G}_{1}}y_{k h} (E) =\sum _{h \in \mathcal{G}_{1}}y_{k h}
\end{equation*}and
\begin{equation*}\sum _{h \in \mathcal{G}_{2}}y_{g h} (E) =\sum _{h \in \mathcal{G}_{2}}y_{g h}
\end{equation*}by the conditions (3) and (4) in definition of $E$. Therefore $s$ is feasible for $l p$ and we get the required contradiction. 
\end{proof}

\begin{lemma}
\label{onecolumn} Let $g$ and $k$ be two different jobs such that $x_{g 1} >0$ and $x_{k 2} >0$. If $\varepsilon _{r} (g ,k) >0$, then no column of type $\binom{ \ast  ,\overline{k}}{ \ast  ,\overline{g}}$ exists in $d (\mathbf{y} ,w)$. 
\end{lemma}

\begin{proof}
If such a column $I =(M_{I} ,\epsilon _{I})$ exists, then $M_{I}$ is $(g ,k)$-feasible semi-matching $E$ in $G$ which contradicts Lemma \ref{reduc}. 
\end{proof}

We now consider another
forbidden configuration of columns in $d (\mathbf{y} ,w)$. Let $I_{1} =(M_{I_{1}} ,\epsilon _{I_{1}})$ and $I_{2} =(M_{I_{2}} ,\epsilon _{I_{2}})$ be two columns. Let $g$, $k$, $a$, and $b$ be four different jobs such that $x_{g 1} >0$, $x_{k 2} >0$, $x_{a 1} >0$, and $x_{b 2} >0$. Define solution $(\mathbf{y} (I_{1} ,I_{2}) ,\mathbf{x} '  ,r'  ,w',\varepsilon )$, where
\begin{equation}\varepsilon  =\min \{\epsilon  ,\varepsilon _{r} (g ,k) ,\varepsilon _{r} (a ,b) ,x_{g 1} ,x_{a 1} ,x_{b 2} ,x_{k 2} ,\epsilon _{I_{1}} ,\epsilon _{I_{2}} ,\min_{j \in \mathcal{J} \setminus D}\{r -\sum _{h}y_{j h}\}\} \label{ax1}
\end{equation}as follows
\begin{equation}y_{j h} (I_{1} ,I_{2}) =\left \{\begin{array}{cc}y_{j h} -\varepsilon  & \;\text{if}\;(j ,h) \in M_{I_{1}}\;\text{and}\;(j ,h) \in M_{I_{2}}\;\text{;}\; \\
y_{j h} -\varepsilon /2 & \;\text{if}\;(j ,h) \in M_{I_{1}}\;\text{and}\;(j ,h) \notin M_{I_{2}}\;\text{;}\; \\
y_{j h} -\varepsilon /2 & \;\text{if}\;(j ,h) \notin M_{I_{1}}\;\text{and}\;(j ,h) \in M_{I_{2}}\;\text{;}\; \\
y_{j h} & \;\text{otherwise}\; ;\end{array}\right . \label{ax2}
\end{equation}
\begin{equation}x' _{j 1} =\left \{\begin{array}{cc}x_{j 1} -\varepsilon /2 & \;\text{if}\;j =g\;\text{or}\;j =a\;\text{;}\; \\
x_{j 1} & \;\text{otherwise}\; ;\end{array}\right . \label{ax3}
\end{equation}
\begin{equation}x'_{j 2} =\left \{\begin{array}{cc}x_{j 2} -\varepsilon /2 & \;\text{if}\;j =k\;\text{or}\;j =b\;\text{;}\; \\
x_{j 2} & \;\text{otherwise}\; ;\end{array}\right . \label{ax4}
\end{equation}

\begin{equation}r' =r -\varepsilon  ; \label{ax5}
\end{equation}

\begin{equation}w'=w -\varepsilon \;\text{.}\; \label{ax6}
\end{equation}

We
have the following lemma

\begin{lemma}
\label{abgkpath} Let $g$, $k$, $a$, and $b$ be four different jobs such that $x_{g 1} >0$, $x_{k 2} >0$, $x_{a 1} >0$, and $x_{b 2} >0$. If $\varepsilon _{r} (g ,k) >0$ and $\varepsilon _{r} (a ,b) >0$, then a column of type $\binom{ \ast  ,a ,b ,g ,k}{ \ast }$ does not exist in $d (\mathbf{y} ,w)$ or a column of type $\binom{ \ast }{ \ast  ,a ,b ,k ,g}$ does not exist in $d (\mathbf{y} ,w)$. 
\end{lemma}

\begin{proof}
Suppose for a contradiction that a column $I_{1}$ of type $\binom{ \ast  ,a ,b ,g ,k}{ \ast }$ is in $d (\mathbf{y} ,w)$ and a column $I_{2}$ of type $\binom{ \ast }{ \ast  ,a ,b ,k ,g}$ is in $d (\mathbf{y} ,w)$. It suffices to show that solution $s =(\mathbf{y} (I_{1} ,I_{2}) ,\mathbf{x}' ,r',w',\varepsilon )$ is feasible for $l p$. The feasibility of $s$ implies that $r$ is not optimal for $l p$ since by the lemma assumptions $\epsilon  >0$ and thus $r' <r$ which leads to a contradiction. To prove the feasibility of $s$ we observe that since both $M_{I_{1}}$ and $M_{I_{2}}$ are matchings that cover all machines we have
\begin{equation*}\sum _{j}y_{j h} (I_{1} ,I_{2}) =\sum _{j}y_{j h} -\varepsilon \;\;\text{\ \ }\;\; \;\;\text{}\;\;h \in \mathcal{G}_{1}
\end{equation*}and
\begin{equation*}\sum _{j}y_{j h} (I_{1} ,I_{2}) =\sum _{j}y_{j h} -\varepsilon \;\;\text{\ \ }\;\; \;\;\text{}\;\;h \in \mathcal{G}_{2}.
\end{equation*}Thus both (\ref{y1}) and (\ref{y2}) hold for $s$ by (\ref{ax5}) and (\ref{ax6}). The (\ref{ax5})
and (\ref{ax6}) also imply (\ref{B20}) for $s$. The (\ref{ax3}) and (\ref{ax4}) imply (\ref{bm1})
and (\ref{bm2}) for $s$ respectively. Since $\varepsilon _{r} (g ,k) >0$ and $\varepsilon _{r} (a ,b) >0$, definition of $\varepsilon $ guarantees (\ref{y3}) for $s$. The $M_{I_{1}} \cup M_{I_{2}}$ covers each job in $D$ exactly twice, thus (\ref{r1d}) holds for each job in $D$ in $s$, and by definition of $\varepsilon $ for each job in $\mathcal{J} \setminus D$. Clearly, the (\ref{y4}) and (\ref{y5}) are met for $s$ by the lemma assumptions and definition of $\varepsilon $. Therefore it remains to show that (\ref{r1a}) and (\ref{r1c})
hold for $s$. First, both $M_{I_{1}}$ and $M_{I_{2}}$ are matchings we have 

\begin{equation*}\sum _{h \in \mathcal{G}_{1}}y_{j h} -\varepsilon  \leq \sum _{h \in \mathcal{G}_{1}}y_{j h} (I_{1} ,I_{2}) \leq \sum _{h \in \mathcal{G}_{1}}y_{j h}\;\;\text{\ \ }\;\; \;\;\text{}\;\;j\in \mathcal{J}
\end{equation*}and
\begin{equation*}\sum _{h \in \mathcal{G}_{2}}y_{j h} -\varepsilon  \leq \sum _{h \in \mathcal{G}_{2}}y_{j h} (I_{1} ,I_{2}) \leq \sum _{h \in \mathcal{G}_{2}}y_{j h}\;\;\text{\ \ }\;\; \;\;\text{}\;\;j\in \mathcal{J}
\end{equation*}Hence (\ref{r1a}) and (\ref{r1c}) hold for all
jobs in $\mathcal{J} \setminus \{a ,g ,b ,k\}$ for $s$. They hold for $g$ and $a$ as well since
\begin{equation*}\sum _{h \in \mathcal{G}_{2}}y_{g h} (I_{1} ,I_{2}) =\sum _{h \in \mathcal{G}_{2}}y_{g h} -\epsilon /2
\end{equation*}and
\begin{equation*}\sum _{h \in \mathcal{G}_{2}}y_{a h} (I_{1} ,I_{2}) =\sum _{h \in \mathcal{G}_{2}}y_{a h} -\epsilon /2
\end{equation*}by the types of $I_{1}$ and $I_{2}$, and for $b$ and $k$ since
\begin{equation*}\sum _{h \in \mathcal{G}_{1}}y_{k h} (I_{1} ,I_{2}) =\sum _{h \in \mathcal{G}_{1}}y_{k h} -\epsilon /2
\end{equation*}and
\begin{equation*}\sum _{h \in \mathcal{G}_{1}}y_{b h} (I_{1} ,I_{2}) =\sum _{h \in \mathcal{G}_{1}}y_{b h} -\epsilon /2
\end{equation*}by the types of $I_{1}$ and $I_{2}$. Therefore $s$ is feasible for $l p$ and we get the required contradiction. 
\end{proof}

The following two corollaries follow immediately from the
proof of Lemma \ref{abgkpath}.

\begin{corollary}
\label{agkpath1} Let $g$, $k$, and $a$ be three different jobs such that $x_{g 1} >0$, $x_{k 2} >0$, and $x_{a 1} x_{a 2} >0$. If $\varepsilon _{r} (g ,a) >0$ and $\varepsilon _{r} (a ,k) >0$, then a column of type $\binom{ \ast  ,a ,g ,k}{ \ast }$ does not exist in $d (\mathbf{y} ,w)$ or a column of type $\binom{ \ast }{ \ast  ,a ,k ,g}$ does not exist in $d (\mathbf{y} ,w)$. 
\end{corollary}

\begin{corollary}
\label{agkpath2} Let $g$, and $k$ be two different jobs such that $x_{g 1} x_{g 2} >0$, and $x_{k 1} x_{k 2} >0$. If $\varepsilon _{r} (g ,k) >0$, then a column of type $\binom{ \ast  ,g ,k}{ \ast }$ does not exist in $d (\mathbf{y} ,w)$ or a column of type $\binom{ \ast }{ \ast  ,k ,g}$ does not exist in $d (\mathbf{y} ,w)$. 
\end{corollary}

\section{Pairs of Columns Absent from $d (\mathbf{y} ,w)$ in $\mathbf{s}$} \label{S5}
Let $g$ and $k$ be two different jobs such that $x_{g 1} >0$, $x_{k 2} >0$.  Let  $I_{k} =(M_{I_{k}} ,\epsilon _{I_{k}})$ be a column of type $\binom{ \ast  ,\overline{k}}{ \ast }$,
and $I_{g} =(M_{I_{g}} ,\epsilon _{I_{g}})$ a column of type $\binom{ \ast }{ \ast  ,\overline{g}}$. Without loss of generality we assume $\epsilon _{I_{k}}= \epsilon _{I_{g}}=\epsilon$. Let $G (I_{g} ,I_{k}) =(M_{I_{g}} \cup M_{I_{k}})$ be a job-machine bi-partite multigraph graph, where an edge connects a machine $h$ and a job $j$ if and only if $(j ,h) \in M_{I_{g}} \cup M_{I_{k}}$. The degree of each machine-vertex in $G (I_{g} ,I_{k})$ is exactly $2$ and the degree of each job-vertex in $G (I_{g} ,I_{k})$ is either $1$ or $2$. Thus, $G (I_{g} ,I_{k})$ is a collection of connected components each of which is either a job-machine path or a job-machine cycle.

\begin{lemma}
\label{twocolumns} If $I_{k} ,I_{g} \in d (\mathbf{y} ,w)$, and $\varepsilon _{r} (g ,k) >0$, then $I_{k}$ is of type $\binom{ \ast }{ \ast  ,k ,g}$ and $I_{g}$ is of type $\binom{ \ast  ,k ,g}{ \ast }$ and both $k$ and $g$ belong to the same connected component of $G (I_{g} ,I_{k})$. 
\end{lemma}

\begin{proof}
By contradiction. We can readily verify that if
\begin{enumerate}
\item $k$ and $g$ are in different connected components of $G (I_{g} ,I_{k})$, or 

\item $k$ or $g$ is missing from $G (I_{g} ,I_{k})$ (i.e. at least one of them is not in $D$), or 

\item $k$ or $g$ is of degree $1$ in $G (I_{g} ,I_{k})$ (i.e. at least one of them is not in $D$), or 

\item $(k ,h) \in G (I_{g} ,I_{k})$ implies $h \in \mathcal{G}_{2}$, or 

\item $(g ,h) \in G (I_{g} ,I_{k})$ implies $h \in \mathcal{G}_{1}$, \end{enumerate}
then there is a matching $M$ of cardinality $m$ in $G (I_{g} ,I_{k}) \subseteq G$ that satisfies definition of $(g ,k)$-feasible semi-matching in $G$. This however contradicts Lemma \ref{reduc} since $I_{k} ,I_{g}$  in $d (\mathbf{y} ,w)$ can be replaced
by columns $I'=(M,\epsilon)$ and $I''= ((M_{I_{g}} \cup M_{I_{k}})\setminus M, \epsilon)$ resulting into another feasible solution to $\ell p$ with the same value $r$ of objective function but with a $(g ,k)$-feasible semi-matching  $M$.
\end{proof}

%\section{Forbidden Columns in $S (\mathbf{y} ,\mathbf{x})$}
%Let $I =(M_{I} ,\epsilon _{I})$ be a column of $S (\mathbf{y} ,\mathbf{x})$. Let $g$ with $x_{g 1} >0$ and $k \in B_{2}$ and $x_{k 2} >0$.
%\begin{equation}\varepsilon _{c} (g ,k) =\sum _{h \in \mathcal{G}_{1}}y_{k h} -(\sum _{h \in \mathcal{G}_{1}}b_{k h} +a_{k 2} -x_{k 2} - \Delta (\mathcal{G}_{2}) +r)
%\end{equation}
%\begin{equation}\varepsilon _{a} (g ,k) =\sum _{h \in \mathcal{G}_{2}}y_{g h} -(\sum _{h \in \mathcal{G}_{2}}b_{g h} +a_{g 1} -x_{g 1} - \Delta (\mathcal{G}_{1}) +r)
%\end{equation}
%\begin{equation}\varepsilon  =\min \{\epsilon  ,\varepsilon _{r} (g ,k) ,\varepsilon _{a} (g ,k) ,\varepsilon _{c} (g ,k) ,x_{g 1} ,x_{k 2} ,\min_{(j ,k) \in E}\{y_{j h}\} ,\min_{j \in \mathcal{J} \setminus D}\{r -\sum _{h}y_{j h}\}\} \label{x1}
%\end{equation}

%\begin{lemma}
%\label{acr} Let $g$ with $x_{g 1} >0$ and $k$ with $x_{k 2} >0$ and $\epsilon _{r} (g ,k) >0$. Then $\min \{\epsilon _{a} (g ,k) ,\epsilon _{c} (g ,k)\} =0$. 
%\end{lemma}

%\begin{proof}
%By contradiction. 
%\end{proof}

%\begin{lemma}
%\label{acr1} Let $g$ with $x_{g 1} >0$ and $k$ with $x_{k 2} >0$ and $\epsilon _{r} (g ,k) >0$. If $\epsilon _{a} (g ,k) >0$, then there is no column of type $\binom{ \ast }{ \ast  ,\overline{g}}$. 
%\end{lemma}

%\begin{proof}
%By contradiction. 
%\end{proof}

%\begin{lemma}
%\label{acr2} Let $g$ with $x_{g 1} >0$ and $k$ with $x_{k 2} >0$ and $\epsilon _{r} (g ,k) >0$. If $\epsilon _{c} (g ,k) >0$, then there is no column of type $\binom{ \ast  ,\overline{k}}{ \ast }$. 
%\end{lemma}

%\begin{proof}
%By contradiction. 
%\end{proof}

\section{The $a$-, $c$-, and $d$-tightness in $\mathbf{s}$} \label{S6}
We show that each job in $B_{1}$ is both $a$-tight and $d$-tight, and each job in $B_{2}$ is both $c$-tight and $d$-tight. We begin by showing the $a$- and $c$- tightness. 

\bigskip 
\begin{lemma}
\label{ac} \label{LL1}Each job $g \in B_{1}$is a-tight and
\begin{equation}\sum _{h \in \mathcal{G}_{2}}y_{g h} <w , \label{ee1}
\end{equation}and each job $k \in B_{2}$ is $c$-tight and
\begin{equation}\sum _{h \in \mathcal{G}_{1}}y_{k h} <w . \label{ee2}
\end{equation}
\end{lemma}

\begin{proof}
By (\ref{B20}), (\ref{r1d}), and (\ref{r1c})
at least one of the following two inequalities
\begin{equation}a_{g 1} -x_{g 1} +\sum _{h \in \mathcal{G}_{2}}(b_{g h} -y_{g h}) < \Delta (\mathcal{G}_{1}) -r \label{ee3}
\end{equation}or
\begin{equation}\sum _{h \in \mathcal{G}_{2}}y_{g h} <w, \label{ee31}
\end{equation}holds for any job $g \in B_{1}$. Likewise, by (\ref{B20}),
(\ref{r1d}), and (\ref{r1a}) at least one of the following two inequalities
\begin{equation}a_{k 2} -x_{k 2} +\sum _{h \in \mathcal{G}_{1}}(b_{k h} -y_{k h}) < \Delta (\mathcal{G}_{2}) -r \label{ee4}
\end{equation}or
\begin{equation}\sum _{h \in \mathcal{G}_{1}}y_{k h} <w\;\text{,}\; \label{ee4aa}
\end{equation}holds for any
job $k \in B_{2}.$ 

Let
$l$ be a job with the largest sum $x_{i 1} +x_{i 2}$ among the jobs $i \in $ $B_{1} \cup B_{2}\;\text{.}\;$ Suppose $l \in B_{1}$ in the proof. If $l \in B_{2}\backslash B_{1}$ the proof proceeds in a similar way and thus will be omitted. 

We prove the lemma for any $k \in B_{2}$ first. For any $k\in B_2$ there is $g\in B_1$ such that $\varepsilon _{r} (g ,k) >0$. Simply take a job $g \in B_{1}$ with the largest $x_{i 1} +x_{i 2}$ among the jobs $i \in B_{1} \setminus \{k\}$,
or take $g=k$ if $B_{1} =\{k\}$. It suffices to show that there is no $i \in (B_{1} \cup B_{2}) \setminus \{g ,k\}$ with $x_{i 1} +x_{i 2} =r$.  Suppose for a contradiction that $x_{i 1} +x_{i 2} =r$ for some $i \in (B_{1} \cup B_{2}) \setminus \{g ,k\}$. If $g \neq k$, then, since $l \in B_{1}$, we have $x_{j 1} +x_{j 2} =r$ for some $j \in \{g ,k\}$. Therefore $\{k ,i ,g\} \subseteq B_{1} \cup B_{2}$ and we get a contradiction by (\ref{bm1}) and (\ref{bm2}). If
$g =k$, then $B_{1} =\{k\}$. 
%Suppose again for a contradiction that $x_{i 1} +x_{i 2} =r$ for some $i \in (B_{1} \cup B_{2}) \setminus \{k\}$. 
Thus, $l =k$ and $x_{k1} +x_{k2} =r$ which implies $B_1 \cup B_{2} =\{i ,k\}$. Since $k\in B_1\cap B_2$ we have $i\in B_1\cap B_2$ which gives contradiction since $i\notin B_1$.

Suppose for a contradiction that $k$ is not $c$-tight, i.e. (\ref{ee4}) holds for $k$. We have $\varepsilon _{c} (k) >0$. If (\ref{ee3}) holds for $g$, then $\varepsilon _{a} (g) >0$. Thus $\min \{\varepsilon _{a} (g) ,\varepsilon _{c} (k)\} >0$ which contradicts Lemma \ref{reduc} since then any column from $d (\mathbf{y} ,w)$ is a $(g ,k)$-feasible matching. If (\ref{ee3}) does not hold
for $g$, then $\varepsilon _{a} (g) =0$ and $\varepsilon _{c} (k) >0$. Take any column $I$ of type $\binom{ \ast }{ \ast  ,\overline{g}}$ from $d (\mathbf{y} ,w)$. This column exists since (\ref{ee31}) holds for $g$. This however contradicts Lemma \ref{reduc} since $I$ is a $(g ,k)$-feasible matching. Therefore, $k$ is $c$-tight and thus by (\ref{ee4aa}) the condition (\ref{ee2})
holds for $k \in B_{2}$. This proves the lemma for any $k \in B_{2}$. 

We now prove the lemma for any $g \in B_{1}$. We begin with $g =l$. There is  $k \in B_{2}$ such that  $\varepsilon _{r} (g ,k) >0$. Simply take a job $k \in B_{2}$ with the largest $x_{i 1} +x_{i 2}$ among the jobs in $B_{2} \setminus \{g\}$, or $k=g$ if $B_{2} =\{g\}$.  Again, it suffices to show that there is no $i \in (B_{1} \cup B_{2}) \setminus \{g,k\}$ such that $x_{i 1} +x_{i 2} =r$. Suppose for a contradiction that $x_{i 1} +x_{i 2} =r$ for some $i \in (B_{1} \cup B_{2}) \setminus \{g ,k\}$. Then, $x_{g 1} +x_{g 2} =r$  implies $B_{1} \cup B_2=\{i, g\}$ by (\ref{bm1}) and (\ref{bm2}).  Hence $k=g$, which implies $B_2=\{g\}$ and $B_1=\{i,g\}$.
Thus, $x_{i 2}$ is integer and by (\ref{bm2}) $x_{g 2} =\lfloor x_{g 2}\rfloor  +\epsilon $. This implies integral $x_{g 1}$ and thus $g \notin B_{1}$ which gives contradiction. 

Suppose for contradiction that $l$ is not $a$-tight, i.e. (\ref{ee3}) holds for $l$. We have $\varepsilon _{a} (l) >0$. If (\ref{ee4}) holds for $k$, then $\varepsilon _{c} (k) >0$. Thus $\min \{\varepsilon _{a} (l) ,\varepsilon _{c} (k)\} >0$ which contradicts Lemma \ref{reduc} since then any column from $d (\mathbf{y} ,w)$ is a $(g ,k)$-feasible matching. If (\ref{ee4}) does not hold
for $k$, then $\varepsilon _{c} (k) =0$ and $\varepsilon _{a} (l) >0$. Take any column $I$ of type $\binom{ \ast  ,\overline{k}}{ \ast }$ from $d (\mathbf{y} ,w)$. This column exists since (\ref{ee4aa}) holds for $k$. This contradicts Lemma \ref{reduc} since $I$ is a $(g ,k)$-feasible matching. Therefore, $l$ is $a$-tight and thus by (\ref{ee31}) the condition (\ref{ee1})
holds for $l \in B_{1}$. This proves the lemma for $l \in B_{1}$. 

To complete the proof assume $B_{1} \setminus \{l\} \neq \emptyset $ in the remainder of the proof. Observe also that for $l$ we have
\begin{equation}\sum _{h \in \mathcal{G}_{1}}y_{l h} <w , \label{ee31abc}
\end{equation}for if
\begin{equation}\sum _{h \in \mathcal{G}_{1}}y_{l h} =w , \label{ee31aa}
\end{equation}then any column $I$ in $d (\mathbf{y} ,w)$ is of type $\binom{ \ast  ,l}{ \ast }$. Since we already have shown that the lemma holds for any $k \in B_{2}$, there is $I$ of type $\binom{ \ast  ,\overline{k} ,l}{ \ast }$ i.e. of type $\binom{ \ast  ,\overline{k}}{ \ast  ,\overline{l}}$.  This contradicts Lemma \ref{onecolumn} since we observe that $\epsilon_r(l,k)$ for  $k\in B_2$ with the largest $x_{k1}+x_{k2}$.

Consider any $g \in B_{1} \setminus \{l\}$. Observe that if $x_{l1}+x_{l2}=r$, then $x_{l2}>0$. Otherwise $B_1=\{l\}$ and we get a contradiction. There is  $k$ such that  $\varepsilon _{r} (g ,k) >0$. Simply take $k =l$, if $x_{l 1} +x_{l 2} =r$, or any $k$  from $B_{2}$ if $x_{l 1} +x_{l 2} <r$. It suffices to show that there is no $i \in (B_{1} \cup B_{2}) \setminus \{g ,k\}$ that satisfies $x_{i 1} +x_{i 2} =r$. Suppose for a contradiction that $x_{i 1} +x_{i 2} =r$ for some $i \in (B_{1} \cup B_{2}) \setminus \{g ,k\}$. Then $x_{k 1} +x_{k 2} =r$. Since $k\neq g$, we have $\{k ,i ,g\} \subseteq B_{1} \cup B_{2}$ and thus we get a contradiction by (\ref{bm1}) and (\ref{bm2}). 

Suppose
for a contradiction that $g$ is not $a$-tight, i.e. (\ref{ee3}) holds for $g$. We have $\varepsilon _{a} (g) >0$. If (\ref{ee4}) holds for $k$, then $\varepsilon _{c} (k) >0$. Thus $\min \{\varepsilon _{a} ,\varepsilon _{c}\} >0$ which contradicts Lemma \ref{reduc} since then any column from $d (\mathbf{y} ,w)$ is a $(g ,k)$-feasible matching. If (\ref{ee4}) does not hold
for $k$, then $\varepsilon _{a} (g) >0$ and $\varepsilon _{c} (k) =0$. Take any column $I$ of type $\binom{ \ast  ,\overline{k}}{ \ast }$ from $d (\mathbf{y} ,w)$. Such column exists for $k \in B_{2}$ since (\ref{ee4aa}) holds in this case, it also exists for $k =l$ (and $l \notin B_{2}$) by (\ref{ee31abc}). This contradicts Lemma \ref{reduc} since
$I$ is a $(g ,k)$-feasible matching. Therefore the lemma holds for each $g \in B_{1}$. 
\end{proof}

We now prove $d$-tightness for each job in $B_{1} \cup B_{2}$.

\begin{theorem}
\label{tightd} \label{t1}Each job in $B_{1} \cup B_{2}$ is $d$-tight. 
\end{theorem}

\begin{proof}
By (\ref{ee1}) in Lemma \ref{ac}, there is a column $I_{g}$ of type $\binom{ \ast }{ \ast  ,\overline{g}}$ in $d (\mathbf{y} ,w)$ for each $g \in B_{1}$. By (\ref{ee2}) in Lemma \ref{ac}, there is a column $I_{k}$ of type $\binom{ \ast  ,\overline{k}}{ \ast }$ in $d (\mathbf{y} ,w)$ for each $k \in B_{2}$. 

Consider job $g$ with the largest $x_{i 1} +x_{i 2}$ among the jobs $i \in $ $B_{1} \cup B_{2}$. Suppose $g \in B_{1}$. If $g \in B_{2}\backslash B_{1}$, then the proof proceeds in a similar way and thus it will be omitted. Take any $k \in B_{2}\setminus\{g\}$ or $k=g$ if $B_2=\{g\}$. Observe that by our choice of $g$, if $x_{i 1} +x_{i 2} =r$ for some $i \in (B_{1} \cup B_{2}) \setminus \{g ,k\}$, then $x_{g 1} +x_{g 2} =r$. Therefore $\{k ,i ,g\} \subseteq B_{1} \cup B_{2}$ which leads to a contradiction by (\ref{bm1}) and (\ref{bm2}) if
$k \neq g$. Otherwise, if $k=g$, then by (\ref{bm2}) $B_1\cup B_2=\{i,g\}$ and $g \in B_1\cap B_2$. Thus $i \in B_1\cap B_2$ and we get contradiction since $i \notin B_{2}$. Thus $\varepsilon _{r} (g ,k) >0$. 

If $k$ is not $d$-tight, then there is a column $I$ of type $\binom{ \ast  ,\overline{k}}{ \ast  ,\overline{k}}$ in $d (\mathbf{y} ,w)$. Thus, if $I \neq I_{g}$, then we get a contradiction with Lemma \ref{twocolumns} applied to $I$ and $I_g$. Otherwise, if $I =I_{g}$, then $I$ is of type of type $\binom{ \ast  ,\overline{k}}{ \ast  ,\overline{g}}$ which contradicts Lemma \ref{onecolumn}.
Similarly, if $g$ is not $d$-tight, then there is a column $I$ of type $\binom{ \ast  ,\overline{g}}{ \ast  ,\overline{g}}$ in $d (\mathbf{y} ,w)$. Thus, if $I \neq I_{k}$, then we get a contradiction with Lemma \ref{twocolumns} applied to $I_k$ and $I$. Otherwise, if $I =I_{k}$, then $I$ is of type of type $\binom{ \ast  ,\overline{k}}{ \ast  ,\overline{g}}$ which contradicts Lemma \ref{onecolumn}.
Therefore the lemma holds for each job in $\{g\} \cup B_{2}$. Moreover, there is a column $I' _{g}$ of type $\binom{ \ast  ,\overline{g}}{ \ast }$. Otherwise all columns in $d (\mathbf{y} ,w)$ are of type $\binom{ \ast  ,g}{ \ast }$ and thus $I_{k}$ is of type $\binom{ \ast  ,\overline{k}}{ \ast  ,\overline{g}}$ for any $k \in B_{2}$ which contradicts Lemma \ref{onecolumn}. 

It remains to prove the lemma
for each $a \in B_{1}\setminus \{g\}$ whenever $B_{1}\setminus \{g\} \neq \emptyset $. Observe that if $x_{g1}+x_{g2}=r$, then $x_{g2}>0$. Otherwise $B_1=\{g\}$ and we get a contradiction. Take a job $k =g$, if $x_{g 1} +x_{g 2} =r$, or any job $k \in B_{2}$, if $x_{g 1} +x_{g 2} <r$. W have $\varepsilon _{r} (a ,k) >0$. This holds since there is no $i \in (B_{1} \cup B_{2}) \setminus \{a ,k\}$ that meets $x_{i 1} +x_{i 2} =r$. Suppose for a contradiction that $x_{i 1} +x_{i 2} =r$ for some $i \in (B_{1} \cup B_{2}) \setminus \{a ,k\}$. Then $x_{k 1} +x_{k 2} =r$. Since $a\neq k$, we have $\{k ,i ,a\} \subseteq B_{1} \cup B_{2}$ which leads to a contradiction by (\ref{bm1}) and (\ref{bm2}). 

Thus
if $a$ is not $d$-tight, then there is a column $I$ of type $\binom{ \ast  ,\overline{a}}{ \ast  ,\overline{a}}$ in $d (\mathbf{y} ,w)$. Then, if  $\varepsilon _{r} (a ,k) >0$ for $k \in B_{2}$, we have either $I \neq I_{k}$ which leads a contradiction with Lemma \ref{twocolumns} applied to $I_k$ and $I$ or $I =I_{k}$ which implies that $I$ is of type $\binom{ \ast  ,\overline{k}}{ \ast  ,\overline{a}}$ which contradicts Lemma \ref{onecolumn}.
If  $\varepsilon _{r} (a ,k) >0$ for $k \notin B_{2}$, then $k =g$. Thus, if $I \neq I' _{g}$, then we get a contradiction with Lemma \ref{twocolumns} applied to $I$ and $I'_g$. Otherwise, if $I =I' _{g}$, then $I$ is of type of type $\binom{ \ast  ,\overline{g}}{ \ast  ,\overline{a}}$ which contradicts Lemma \ref{onecolumn}.
\end{proof}

\bigskip For $j \in B_{1}\cup B_2$ define
\begin{equation*}\alpha _{j} =\sum _{h \in \mathcal{G}_{1}}y_{j h}\;\;\text{and}\;\;\beta _{j} =\sum _{h \in \mathcal{G}_{2}}y_{j h}.
\end{equation*}
%and
%\begin{equation*}\beta _{k} =\sum _{h \in \mathcal{G}_{2}}y_{k h}\;\;\text{and}\;\;\beta _{g} =\sum _{h \in \mathcal{G}_{2}}y_{g h}.
%\end{equation*}

The following two lemmas relate the fractions of $x_{j 1}$, $x_{j 2}$, $\alpha _{j}$, and $\beta _{j}$ for $j \in B_{1} \cup B_{2}$. The lemmas follow from Lemmas \ref{LL1}, and Theorem \ref{t1} and
will prove useful in the remainder of the paper.
\begin{lemma}
\label{l6}For $g\in B_1$, let
\begin{equation*}x_{g 1} =\left \lfloor x_{g 1}\right \rfloor  +\epsilon _{g}\;\;\text{,}\;\;\beta _{g} =\left \lfloor \beta _{g}\right \rfloor  +\lambda _{g} ,\;\;\text{and}\;\;\alpha _{g} =\left \lfloor \alpha _{g}\right \rfloor  +\omega _{g}
\end{equation*}where $0 \leq \lambda _{g}$, $\omega _{g} <1\;\text{,}\;$$0 <\epsilon _{g} <1$ for $g \in B_{1}$. Then, $\omega_g=\varepsilon_g$, and $\lambda_g=\varepsilon-\varepsilon_g$ for $\varepsilon\geq \varepsilon_g$, and
$\lambda_g=1-(\varepsilon_g - \varepsilon)$ for $\varepsilon< \varepsilon_g$.
%\begin{equation*}\epsilon _{g} +\lambda _{g} -\epsilon \;\;\text{and}\;\;\lambda _{g} +\omega _{g} -\epsilon 
%\end{equation*}are integral. 
\end{lemma}

\begin{proof}
By Lemma \ref{ac} $g \in B_{1}$ is $a$-tight, i.e.
\begin{equation*}a_{g 1} -x_{g 1} +\sum _{h \in \mathcal{G}_{2}}b_{g h} -\beta _{g} = \Delta (\mathcal{G}_{1}) -r,
\end{equation*}by Theorem \ref{t1}
$g$ is $d$-tight, i.e.
\begin{equation}\label{W0}
\alpha _{g} +\beta _{g} =w.
\end{equation}
The two imply $\omega_g=\varepsilon_g$ since $w-r$ is integral. By (\ref{W0}), $\lambda_g + \omega_g - \varepsilon=0$ or $1$. Thus $\lambda_g + \varepsilon_g - \varepsilon=0$ or $1$. Therefore,
$\lambda_g =\varepsilon -\varepsilon_g$ for $\varepsilon\geq \varepsilon_g$, and 
$\lambda_g=1-(\varepsilon_g - \varepsilon)$ for $\varepsilon< \varepsilon_g$.
\end{proof}

\begin{lemma}
\label{l7}For $j\in B_2$, let
\begin{equation*}x_{k 2} =\left \lfloor x_{k 2}\right \rfloor  +\epsilon _{k}\;\;\text{and}\;\;\beta _{k} =\left \lfloor \beta _{k}\right \rfloor  +\lambda _{k}\;\;\text{and}\;\;\alpha _{k} =\left \lfloor \alpha _{k}\right \rfloor  +\omega _{k}
\end{equation*}where $0 \leq \lambda _{k}$, $\omega _{k} <1\;\text{,}\;$ $0 <\epsilon _{k} <1$ for a job $k \in B_{2}$. Then,  $\lambda_k=\varepsilon_k$, and $\omega_k=\varepsilon-\varepsilon_k$ for $\varepsilon\geq \varepsilon_k$, and
$\lambda_k=1-(\varepsilon_k - \varepsilon)$ for $\varepsilon< \varepsilon_k$.
%\begin{equation*}\epsilon _{k} +\omega _{k} -\epsilon \;\;\text{and}\;\;\lambda _{k} +\omega _{k} -\epsilon 
%\end{equation*}are integral. 
\end{lemma}

\begin{proof} The proof is similar to the proof of Lemma \ref{l6} and will be omitted.
%By Lemma \ref{LL1} $k \in B_{2}$ is $c$-tight, i.e.
%\begin{equation*}a_{k 2} -x_{k 2} +\sum _{h \in \mathcal{G}_{1}}b_{g h} -\alpha _{k} = \Delta (\mathcal{G}_{2}) -r\;\text{,}\;
%\end{equation*}and by Theorem
%\ref{t1} $k$ is $d$-tight, i.e.
%\begin{equation*}\alpha _{k} +\beta _{k} =w\;\text{.}\;
%\end{equation*}
\end{proof}

\section{The Absence of Crossing Jobs in $\mathbf{s}$} \label{crossi}

%In this section we prove that $\varepsilon _{r} (g ,k) >0$ for each pair $g \in B_{1}$ and $k \in B_{2}$. To that end we introduce the following two definitions. 
Each job $k \in B_{1} \cap B_{2}$ is called \emph{crossing}. We call a job $a \in B_{1} \cup B_{2}$ an \emph{e-crossing} job, if it meets the following conditions:

\begin{itemize}
\item $0 <x_{a 2}$ and $0 <x_{a 1}$; 

\item both $B_{1} \setminus \{a\}$ and $B_{2} \setminus \{a\}$ are not empty. \end{itemize}

\noindent We have the following.

\begin{theorem}
\label{cnc} Each crossing job is $e$-crossing. 
\end{theorem}

\begin{proof}
Suppose for a contradiction that $a$ is crossing but not $e$-crossing. By Theorem \ref{tightd} $a$ is $d$-tight and thus
\begin{equation}\sum _{h \in \mathcal{G}_{2}}y_{a h} +\sum _{h \in \mathcal{G}_{1}}y_{a h} =w . \label{ee31a}
\end{equation}By Lemma \ref{ac} $a$ is both $a$-tight and $c$-tight, thus 

\begin{equation}a_{a 1} -x_{a 1} +\sum _{h \in \mathcal{G}_{2}}(b_{a h} -y_{a h}) = \Delta (\mathcal{G}_{1}) -r \label{ee3a}
\end{equation}and
\begin{equation}a_{a 2} -x_{a 2} +\sum _{h \in \mathcal{G}_{1}}(b_{a h} -y_{a h}) = \Delta (\mathcal{G}_{2}) -r . \label{ee4a}
\end{equation}By summing up (\ref{ee31a}), (\ref{ee3a}),
and (\ref{ee4a}) side by side we obtain 

\begin{equation}a_{a 1} +a_{a 2} +\sum _{h}b_{a h} - \Delta (\mathcal{G}_{1}) - \Delta (\mathcal{G}_{2}) +r -w = -r +x_{a 1} +x_{a 2}\;\text{.}\; \label{ee3aaa}
\end{equation}Since $a$ is not $e$-crossing, $B_{1} \setminus \{a\} =\emptyset $ or $B_{2} \setminus \{a\} =\emptyset $. Thus, $x_{a 1} =\lfloor x_{a 1}\rfloor  +\epsilon $ or $x_{a 2} =\lfloor x_{a 2}\rfloor  +\epsilon $. Therefore, the left hand side of (\ref{ee3aaa}) is integral but its right hand side is fractional
since both $x_{a 1}$ and $x_{a 2}$ are fractional. This leads to contradiction and thus the theorem holds. 
\end{proof}

\begin{theorem}
\label{T4A} For each $e$-crossing job $a$ we have $x_{a 1} +x_{a 2} <r$. 
\end{theorem}

\begin{proof}
By contradiction. Let $a$ be $e$-crossing with $x_{a 1} +x_{a 2} =r$. Let $g \in B_{1} \setminus \{a\}$ and $k \in B_{2} \setminus \{a\}$. By Theorem \ref{t1} and Lemma \ref{ac} there are columns $I_{k}$ of type $\binom{ \ast }{ \ast  ,k}$ and $I_{g}$ of type $\binom{ \ast  ,g}{ \ast }$ in $d (\mathbf{y} ,w)$. By Theorem \ref{t1} $I_{k}$ is either of type $\binom{ \ast }{ \ast  ,k ,g}$ or of type $\binom{ \ast  ,g}{ \ast  ,k}$, and $I_{g}$ is either of type $\binom{ \ast  ,g ,k}{ \ast }$ or of type $\binom{ \ast  ,g}{ \ast  ,k}$. Suppose that $I_{k}$ or $I_{g}$ is of type $\binom{ \ast  ,g}{ \ast  ,k}$, then $g \neq k$. Since $a$ is $e$-crossing, by Theorem \ref{t1} this column, say $I$, is either of type $\binom{ \ast  ,a ,g}{ \ast  ,k}$ or of type $\binom{ \ast  ,g}{ \ast  ,a ,k}$. The former is of type $\binom{ \ast  ,\overline{k}}{ \ast  ,\overline{a}}$ and the latter of type $\binom{ \ast  ,\overline{a}}{ \ast  ,\overline{g}}$. Since $g \neq k$, $a$ is the only job $i$ with $x_{i 1} +x_{i 2} =r$. Thus $\varepsilon _{r} (a ,k) >0$ and $\varepsilon _{r} (g ,a) >0$. Therefore we get a contradiction with Lemma \ref{onecolumn} which implies that $I_{g}$ is of type $\binom{ \ast  ,g ,k}{ \ast }$ and $I_{k}$ is of type $\binom{ \ast }{ \ast  ,k ,g}$ (observe that we may now have $g =k$). Since $a$ is $e$-crossing, by Theorem \ref{t1} we have $I_{g}$ of type $\binom{ \ast  ,a ,g ,k}{ \ast }$ or of type $\binom{ \ast  ,g ,k}{ \ast  ,a}$, and $I_{k}$ is of type $\binom{ \ast }{ \ast  ,a ,k ,g}$ or of type $\binom{ \ast  ,a}{ \ast  ,k ,g}$. The $I_{g}$ of type $\binom{ \ast  ,g ,k}{ \ast  ,a}$ is of type $\binom{ \ast  ,\overline{a}}{ \ast  ,\overline{g}}$, and the $I_{k}$ of type $\binom{ \ast  ,a}{ \ast  ,k ,g}$ is of type $\binom{ \ast  ,\overline{k}}{ \ast  ,\overline{a}}$. Moreover, if $g \neq k$, then $a$ is the only job $i$ with $x_{i 1} +x_{i 2} =r$, and if $k =g$, then either $x_{k 1} +x_{k 2} =r$ or $a$ is the only job $i$ with $x_{i 1} +x_{i 2} =r$. Thus $\varepsilon _{r} (a ,k) >0$ and $\varepsilon _{r} (g ,a) >0$. Therefore, $I_{g}$ being of type $\binom{ \ast  ,g ,k}{ \ast  ,a}$ or $I_{k}$  being of type $\binom{ \ast  ,a}{ \ast  ,k ,g}$ contradicts Lemma \ref{onecolumn}. Thus it
remains to consider $I_{g}$ of type $\binom{ \ast  ,a ,g ,k}{ \ast }$ and $I_{k}$ is of type $\binom{ \ast }{ \ast  ,a ,k ,g}$. This leads to a contradiction by Corollaries \ref{agkpath1}
and \ref{agkpath2} since $\varepsilon _{r} (g ,a) >$ and $\varepsilon _{r} (a ,k) >0$. The last two inequalities clearly hold if $a$ is the only job $i$ with $x_{i 1} +x_{i 2} =r$, otherwise $g =k$ and $k$ is the other job $i$ with $x_{i 1} +x_{i 2} =r$. 
\end{proof}

The following corollary follows immediately from the proof of Theorem \ref{T4A} since the assumption $x_{i 1} +x_{i 2} <r$ for each $i \in B_{1} \cup B_{2}$ implies  $\varepsilon_r(g,k)>0$ for each $g\in B_1$ and $k\in B_2$.

\begin{corollary}
\label{col} If $x_{i 1} +x_{i 2} <r$ for each $i \in B_{1} \cup B_{2}$, then no job is $e$-crossing. 
\end{corollary}

\noindent We are now ready to prove two main results
of this section.

\begin{theorem}
\label{T4} No crossing job exists. 
\end{theorem}

\begin{proof}
By contradiction. Suppose $a$ is a crossing job. Take a crossing job with the largest $x_{a 1} +x_{a 2}$. By Theorem \ref{cnc} $a$ is $e$-crossing, and by Theorem \ref{T4A} $x_{a 1} +x_{a 2} <r$. By Corollary \ref{col} $x_{i 1} +x_{i 2} =r$ for some $i \in B_{1} \cup B_{2}$. Thus $i \neq a$. By Theorem \ref{T4A} $i$ is not $e$-crossing. Thus ($x_{i 1} =0$ or $x_{i 2} =0$)  which implies ( $B_{1} =\{i\}$ or $B_{2} =\{i\}$). This leads to contradiction since $a \in B_{1} \cap B_{2}$ and $a \neq i$. 
\end{proof}

\begin{theorem}
\label{eps} For each $g \in B_{1}$ and $k \in B_{2}$, $\varepsilon _{r} (g ,k) >0$. 
\end{theorem}

\begin{proof}
Suppose for a contradiction that $\varepsilon _{r} (g ,k) =0$ for some $g \in B_{1}$ and $k \in B_{2}$. By Theorem \ref{T4}, $g \neq k$. Then $r =x_{j 1} +x_{j 2}$ for some $j \in (B_{1} \cup B_{2}) \setminus \{g ,k\}$. By Theorem \ref{T4} $j$ is not crossing thus $\{j ,g\} \subseteq B_{1}$ and $j \notin B_{2}$, or $\{j ,k\} \subseteq B_{2}$ and $j \notin B_{1}$. Suppose the former, the proof for the latter is similar and thus will be omitted. We have $x_{j 2}$ integral. However, by Theorem \ref{T4A} $j$ is not $e$-crossing. Hence $x_{j 2} =0$. Thus $r =x_{j 1}$ and $B_{1} =\{j\}$ which gives a contradiction. 
\end{proof}

\section{Characterization of $d (\mathbf{y} ,w)$ in $\mathbf{s}$} \label{S8}
We give a characterization of $d (\mathbf{y} ,w)$ that will be used in the remainder of the proof.

\begin{lemma}\label{D}
\label{b1b2} For each $g \in B_{1}$ and $k \in B_{2}\;\text{,}\;$ any column $I$ in $d (\mathbf{y} ,w)$ is either of type $\binom{ \ast  ,k}{ \ast  ,g}$ or of type $\binom{ \ast }{ \ast  ,k ,g}$ or of type $\binom{ \ast  ,k ,g}{ \ast }\;\text{.}\;$ Moreover, for each $g \in B_{1}$ and $k \in B_{2}\;\text{,}\;$ there is $I_{k}$ of type $\binom{ \ast }{ \ast  ,k ,g}\;\text{,}\;$ and there is $I_{g}$ of type $\binom{ \ast  ,k ,g}{ \ast }$ in $d (\mathbf{y} ,w)$. Finally, if there is $i \in B_{1} \cup B_{2}$ such that $x_{i 1} +x_{i 2} =r$, then either $B_{1} =\{i\}$ or $B_{2} =\{i\}$. 
\end{lemma}

\begin{proof} Let $g \in B_{1}$ and $k \in B_{2}$.
By Lemma \ref{ac} and Theorem \ref{t1} there are columns $I_{k}$ of type $\binom{ \ast }{ \ast  ,k}$ and $I_{g}$ of type $\binom{ \ast  ,g}{ \ast }$ in $d (\mathbf{y} ,w)$. By Theorem \ref{t1} $I_{k}$ is either of type $\binom{ \ast }{ \ast  ,k ,g}$ or of type $\binom{ \ast  ,g}{ \ast  ,k}$, and $I_{g}$ is either of type $\binom{ \ast  ,g ,k}{ \ast }$ or of type $\binom{ \ast  ,g}{ \ast  ,k}$. By Theorem \ref{eps} we have  $\varepsilon _{r} (g ,k) >0$ and thus by Lemma \ref{onecolumn} neither $I_{k}$ nor $I_{g}$ is of type $\binom{ \ast  ,g}{ \ast  ,k}$. This proves the first part of the lemma. 

If there is $i \in B_{1} \cup B_{2}$ such that $x_{i 1} +x_{x 2} =r$, then by Theorem \ref{T4} the $i$ is not crossing. Hence either $i \in B_{1} \setminus B_{2}$ or $i \in B_{2} \setminus B_{1}$. By Theorem \ref{T4A} $i$ is not $e$-crossing thus either $B_{1} =\{i\}$ or $B_{2} =\{i\}$. 
\end{proof}

\begin{theorem}
\label{product} If there is a job $j$ such that $x_{j 1} x_{j 2} >0$, then $B_{1} =\{j\}$ or $B_{2} =\{j\}$. 
\end{theorem}

\begin{proof}
Let $x_{j 1} x_{j 2} >0$ for a job $j$. Without loss of generality let $j$ be a job with the largest value of $x_{j1}+x_{j2}$ among jobs with $x_{j 1} x_{j 2} >0$. Suppose for a contradiction that $B_{1} \setminus \{j\} \neq \emptyset $ and $B_{2} \setminus \{j\} \neq \emptyset $. Thus if $j\in B_1\cup B_2$, then $j$ is $e$-crossing. By Theorem \ref{T4A}, $x_{j 1} +x_{j 2} <r$. Let $g \in B_{1} \setminus \{j\}$ and $k \in B_{2} \setminus \{j\}$. By Theorem \ref{eps} we have $\epsilon _{r} (g ,j) >0$ and $\epsilon _{r} (j ,k) >0$.  Thus by Corollary \ref{agkpath1}
column of type $\binom{ \ast  ,j ,g ,k}{ \ast }$ does not exist in $d (\mathbf{y} ,w)$ or column of type $\binom{ \ast }{ \ast  ,j ,k ,g}$ does not exist in $d (\mathbf{y} ,w)$. Suppose the former holds, then by Lemma \ref{b1b2} a column of type $\binom{ \ast  ,\bar{j}}{ \ast  ,\bar{g}}$ exists in $d (\mathbf{y} ,w)$ which contradicts Lemma \ref{onecolumn}. For the latter, by Lemma \ref{b1b2}
a column of type $\binom{ \ast  ,\bar{k}}{ \ast  ,\bar{j}}$ exists in $d (\mathbf{y} ,w)$ which contradicts Lemma \ref{onecolumn}. 

If $j\notin B_1 \cup B_2$, then both $x_{j1}$ and $x_{j2}$ are integral. Thus $x_{j1} + x_{j2}< r$.  Let $g \in B_{1} \setminus \{j\}$ and $k \in B_{2} \setminus \{j\}$. We have $\epsilon _{r} (g ,j) >0$ and $\epsilon _{r} (j ,k) >0$.
To prove the former inequality we observe that by our choice of job $j$ for any job $i\in B_1 \cup B_2$ different from $g$ and $j$,  and such that $x_{i1}+x_{i2}=r$ must be either $r =x_{i 1}$ or $r =x_{i 2}$. Otherwise
$x_{i1}x_{i2}>0$, thus $i$ would have been chosen instead of $j$. The proof of the latter inequality follows by a similar argument.  Thus by Corollary \ref{agkpath1}
column of type $\binom{ \ast  ,j ,g ,k}{ \ast }$ does not exist in $d (\mathbf{y} ,w)$ or column of type $\binom{ \ast }{ \ast  ,j ,k ,g}$ does not exist in $d (\mathbf{y} ,w)$. Suppose the former holds, then by Lemma \ref{b1b2} a column of type  $\binom{ \ast , g ,k}{ \ast }$   exists in $d (\mathbf{y} ,w)$. This column is either of type $\binom{ \ast , g ,k}{ \ast, j }$  or of type $\binom{ \ast , \bar{j}, g ,k}{ \ast. \bar{j} }$ which implies that the column is of type $\binom{ \ast  ,\bar{j}}{ \ast  ,\bar{g}}$. This however contradicts Lemma \ref{onecolumn}. For the latter, we prove in a similar fashion that 
a column of type $\binom{ \ast  ,\bar{k}}{ \ast  ,\bar{j}}$ exists in $d (\mathbf{y} ,w)$ which contradicts Lemma \ref{onecolumn}. 
Therefore we get a contradiction which proves the theorem.
\end{proof}

\subsection{The overlap} \label{S9}
An overlap of $B_{1}$ is a column $I=(M_I, \epsilon)\in d (\mathbf{y} ,w)$ that matches at least two different jobs from $B_{1}$ with machines in $\mathcal{G}_{1}$. Similarly, an overlap of $B_{2}$ is a column $I=(M_I, \epsilon) \in d (\mathbf{y} ,w)$ that matches at least two different jobs from $B_{2}$ with machines in $\mathcal{G}_{2}$.

\begin{lemma}
\label{overlap} An overlap of $B_{1}$ and an overlap of $B_{2}$  do not occur simultaneously. 
\end{lemma}

\begin{proof}
Suppose for contradiction that both overlaps occur simultaneously. Then there are different jobs $a$ and $g$ both from $B_{1}$ done on $\mathcal{G}_{1}$ in a column $I_{a ,g} \in d (\mathbf{y} ,w)$ of type $\binom{ \ast  ,a ,g}{ \ast }$, and different jobs $b$ and $k$ both from $B_{2}$ done on $\mathcal{G}_{2}$ in a column $I_{b ,k} \in d (\mathbf{y} ,w)$ of type $\binom{ \ast }{ \ast  ,b ,k}$. By Lemma \ref{b1b2}, $I_{a ,g}$ is of type $\binom{ \ast  ,a ,b ,k ,g}{ \ast }$ and $I_{b ,k}$ is of type $\binom{ \ast }{ \ast  ,a ,b ,k ,g}$. This, by Theorem \ref{eps}, contradicts Lemma
\ref{abgkpath}  (by Theorem \ref{T4} there are no crossing jobs thus all four jobs $a$, $g$, $b$, and $k$ are different). This proves the lemma.
\end{proof}

%\end{comment}

\section{Integral optimal solution to $\ell p$ for $\sum_{j\in B_1} \varepsilon_j=\epsilon$ or $ \sum_{j\in B_2} \varepsilon_j=\epsilon$.} \label{S10}

In this section we prove that an integral optimal solution for $\ell p$ exists if $\epsilon>0$ and  $\sum_{j\in B_1} \varepsilon_j=\epsilon$ or $ \sum_{j\in B_2} \varepsilon_j=\epsilon$. We first prove this assuming $\sum_{j\in B_2} \varepsilon_j=\epsilon$ in $\mathbf{s}$ throughout this section. The proof for  $\sum_{j\in B_1} \varepsilon_j=\epsilon$  proceeds in a similar fashion and thus will be omitted. 

Consider the following network flow problem $\cal{F}$ with variables $t_{jh}$ for $j$ and $h\in \mathcal{G}_2$, and $z_{jh}$ for $j$ and $h\in \mathcal{G}_1$. The $r$, $w$, and $x_{j\ell}$ for $j\in \mathcal{J}$ and $\ell=1,2$ in $\cal{F}$ are constants obtained from the solution $\mathbf{s}=(\mathbf{y},\mathbf{x},r,w)$.
\bigskip

%$F$

\begin{equation*}
F=\max \sum\limits_{j\in B_1} \sum\limits_{h\in \mathcal{G}_2}t_{jh}
\end{equation*}%
Subject to%
\begin{equation}
\sum\limits_{j}t_{jh}=\lfloor w\ \rfloor \text{ \ \ } \text{\ }h\in \mathcal{G}_{2}
\label{AB2af}
\end{equation}%

\begin{equation}
\sum\limits_{h\in \mathcal{G}_{2}}b_{jh}+a_{j1}- \Delta (%
\mathcal{G}_{1})+\lfloor r \rfloor - x_{j1}\leq  \sum\limits_{h\in \mathcal{G}_{2}}t_{jh} \text{ \ \ \ }\text{\ }j\in \mathcal{J}\setminus B_1  \label{AB10af}
\end{equation}%

\begin{equation}
\sum\limits_{h\in \mathcal{G}_{2}}b_{jh}+a_{j1}- \Delta (%
\mathcal{G}_{1})+\lfloor r \rfloor - \lceil x_{j1} \rceil \leq  
 \sum\limits_{h\in \mathcal{G}_{2}}t_{jh}\leq \sum\limits_{h\in \mathcal{G}_{2}}b_{jh}+a_{j1}- \Delta (%
\mathcal{G}_{1})+ \lfloor r \rfloor - \lfloor x_{j1} \rfloor \text{ \ \ \ } \text{\ }j\in B_1 \label{AB10affz}
\end{equation}%

\begin{equation}
\sum\limits_{j}z_{jh}= \lfloor w \rfloor \text{ \ \ \ } \text{ }h\in \mathcal{G}%
_{1}  \label{AB3af}
\end{equation}%

\begin{equation}
\sum\limits_{h\in \mathcal{G}_{1}}b_{jh}+a_{j2} - \Delta (%
\mathcal{G}_{2})+ \lfloor r  \rfloor - \lfloor x_{j2} \rfloor \leq  \sum\limits_{h\in \mathcal{G}_{1}}z_{jh} \text{ \ \ \ }\text{\ }j\in \mathcal{J}  \label{AB10aafk}
\end{equation}%

\begin{equation}
0\leq t_{jh}\leq b_{jh} \;\;\text{\ }\;\; \;\;\text{}\;\;h \in \mathcal{M}\;\; \;\;\text{\ }\;\;j \in \mathcal{G}_2 \label{AB5af}
\end{equation}%
\begin{equation}
0\leq z_{jh}\leq b_{jh} \;\;\text{\ }\;\; \;\;\text{}\;\;h \in \mathcal{M}\;\; \;\;\text{\ }\;\;j \in \mathcal{G}_1 \label{AB5aaf}
\end{equation}%
\begin{equation}
\sum\limits_{h\in \mathcal{G}_1}z_{jh} + \sum\limits_{h\in \mathcal{G}_2}t_{jh} \leq \lfloor w \rfloor \text{ \ }\text{\ }j\in \mathcal{J}  \label{AB4afq}
\end{equation}%

\begin{lemma} \label{lowerb1}
There is a feasible solution to $\cal{F}$ with value 
\begin{equation}\label{value}
 \sum\limits_{j\in B_1} \sum\limits_{h\in \mathcal{G}_2}b_{jh} - \sum\limits_{j\in \mathcal{J}\setminus B_1} (a_{j1}-x_{j1})-(|B_1|-1)(\Delta(\mathcal{G}_1)-\lfloor r \rfloor)-\epsilon.
\end{equation}
\end{lemma}
\begin{proof}
For $\mathbf{s}$, consider the set $Y_j$ of all columns  of type $\binom {%
\ast }{\ast ,j}$ in $d(\mathbf{y},w)$ for $j\in B_2$.
By Lemma \ref{l7},  $l(Y_j)=\beta_j=\lfloor \beta_j \rfloor + \varepsilon_j$. If there is no overlap of $B_2$ or $\sum_{j\in B_2} \lfloor \beta_j \rfloor >0$, then
take an interval $Y \subseteq \bigcup_{j\in B_2} Y_j $ such that $l(Y)= \epsilon$, $l(Y\cap Y_j)\geq \varepsilon_j$ for $j\in B_2$. Otherwise, if there is overlap of $B_2$ and $\sum_{j\in B_2} \lfloor \beta_j \rfloor =0$, then take an interval
$Y \subseteq (\bigcup_{j\in B_2} Y_j)\cup Z $ such that $l(Y)= \epsilon$, $l(Y\cap Y_j)\geq \varepsilon_j$ for $j\in B_2$. Here the $Z$ is  the set  of all columns  of type $\binom {%
\ast,B_2 }{\ast , B_1}$ in $d(\mathbf{y},w)$. In order for such $Y$ to exist we  show that $l( (\bigcup_{j\in B_2} Y_j)\cup Z)\geq 1$. 
%First, we have $0<l( \bigcup_{j\in B_2} Y_j)=\varepsilon'<\varepsilon$. Second, 
By Lemma \ref{overlap} there is no overlap of $B_1$, thus
$l( \bigcup_{j\in B_1} W_j)=\epsilon + i$ for some integer $i\geq 0$, where  $W_j$ is the set of all columns  of type $\binom {%
\ast, j}{\ast }$ in $d(\mathbf{y},w)$ for $j\in B_1$. Thus $l(d(\mathbf{y},w)\setminus  \bigcup_{j\in B_1} W_j)$ is integral since $l(d(\mathbf{y},w))=w$, and positive. However $d(\mathbf{y},w)\setminus  \bigcup_{j\in B_1} W_j)=(\bigcup_{j\in B_2} Y_j)\cup Z$
by Theorem \ref{tightd} and Lemma \ref{D}. This proves $l( (\bigcup_{j\in B_2} Y_j)\cup Z)\geq 1$, and the required $Y$ exists.

Let $Y_{jh}$ be the set of columns $I\in Y$ such that $(j,h)\in M_I$, set $\gamma_{jh}:=l(Z_{jh})$. Informally, $\gamma_{jh}$ is the amount of $j$ done on $h$ in the interval $Y$. We define a truncated solution as follows $z^*_{jh}:=y_{jh}-\gamma_{jh}$ for $h\in \mathcal{G}_1$, and  $t^*_{jh}:=y_{jh}-\gamma_{jh}$ for $h\in \mathcal{G}_2$.
By Theorem \ref{tightd} each $j\in B_2$ is $d$-tight
% and no overlap on $B_2$ 
 thus
\begin{equation}
%\sum\limits_{h\in \mathcal{G}_{2}}b_{jh}+a_{j1}- \Delta (%
%\mathcal{G}_{1})+\lfloor r - x_{j1}\rfloor \leq  
 \sum\limits_{h\in \mathcal{G}_{1}}\gamma_{jh}+  \sum\limits_{h\in \mathcal{G}_{2}}\gamma_{jh}=\epsilon \text{ \ \ \ } \text{\ }j\in  B_2   \label{B10affzzz}
\end{equation}%
and 
\begin{equation} \label{CX1}
\sum\limits_{h\in \mathcal{G}_{2}}\gamma_{jh}=\eta_j\geq \varepsilon_j \text{ \ \ \ } \text{\ }j\in  B_2.
\end{equation}
We prove that this truncated solution is feasible for $\cal{F}$ and meets (\ref{value}).

We first prove the following  lemma.

\begin{lemma} \label{Abothsmaller}
If $ \sum\limits_{j\in B_2}\varepsilon_{j}=\epsilon$,  then truncated solution meets (\ref{AB10aafk}).   
\end{lemma}

\begin{proof}%This solution truncates $\mathbf{y}$ to $\mathbf{t}^*$ on $\mathcal{G}_2$ and to $\mathbf{z}^*$ on $\mathcal{G}_1$. 
We have the following for the truncated solution.

\begin{equation} \label{AyS2}
\sum\limits_{h\in \mathcal{G}_1} z^*_{jh}=\sum\limits_{h\in \mathcal{G}_1} y_{jh} -(\epsilon-\eta_j) \text{ \ \ \ } \text{\ }j\in  B_2.
%\text {\ \ \ and} \text {\ \ \ }  \sum\limits_{h\in \mathcal{G}_1} z^*_{bh}= \sum\limits_{h\in \mathcal{G}_1} y_{bh}-(\epsilon-\epsilon_b). 
\end{equation}
By Lemma \ref{ac} each $j\in B_2$ is $c$-tight. Thus 
%since $\epsilon=\epsilon_k+\epsilon_b$, 
we get
\begin{equation}\label{AS}
\sum\limits_{h\in \mathcal{G}_1} y_{jh}=\sum\limits_{h\in \mathcal{G}_{1}}b_{jh}+a_{j2} - \Delta (\mathcal{G}_{2}) - \lfloor x_{j2} \rfloor + \lfloor r \rfloor +\epsilon -\varepsilon_j \text{ \ \ \ } \text{\ }j\in  B_2.
\end{equation}
%and
%\begin{equation}
%\sum\limits_{h\in \mathcal{G}_1} y_{bh}=\sum\limits_{h\in \mathcal{G}_{1}}b_{bh}+a_{b2} - \Delta (\mathcal{G}_{2}) - \lfloor x_{b2} \rfloor + \lfloor r \rfloor +  \epsilon_k.
%\end{equation}
Therefore by (\ref{AyS2})  and (\ref{AS}) 
we get 
\begin{equation*}
\sum\limits_{h\in \mathcal{G}_1} z^*_{jh} + (\varepsilon_j - \eta_j)= \sum\limits_{h\in \mathcal{G}_{1}}b_{jh}+a_{j2} - \Delta (\mathcal{G}_{2}) - \lfloor x_{j2} \rfloor + \lfloor r \rfloor \text{ \ \ \ } \text{\ }j\in  B_2,
\end{equation*}
and by (\ref{CX1})
\begin{equation*}
\sum\limits_{h\in \mathcal{G}_1} z^*_{jh} \geq \sum\limits_{h\in \mathcal{G}_{1}}b_{jh}+a_{j2} - \Delta (\mathcal{G}_{2}) - \lfloor x_{j2} \rfloor + \lfloor r \rfloor \text{ \ \ \ } \text{\ }j\in  B_2,
\end{equation*}
which proves (\ref{AB10aafk}) holds for the truncated solution $\mathbf{t}^*$  and  $\mathbf{z}^*$. 

\end{proof}

Let $\mathbf{t}^*$ and $\mathbf{z}^*$ be a solution of Lemma \ref{Abothsmaller}.
The   $\mathbf{t}^*$ and $\mathbf{z}^*$ clearly meet
(\ref{AB2af}), (\ref{AB3af}), (\ref{AB5af}), (\ref{AB5aaf}),  (\ref{AB4afq}). By Lemma 
%\ref{bothgreater} and 
\ref{Abothsmaller} (\ref{AB10aafk}) holds. Then  (\ref{AB10af})
%, the left hand side inequality in (\ref{B10aff})
 also holds for $\mathbf{t}^*$ and $\mathbf{z}^*$. To show that we observe that 
 by feasibility of $\mathbf{s}=(\mathbf{y},\mathbf{x},r,w)$ we have
\begin{equation*}
\sum\limits_{h\in \mathcal{G}_{2}}b_{jh}+a_{j1}- x_{j1}- \Delta (%
\mathcal{G}_{1})+r  \leq  \sum\limits_{h\in \mathcal{G}_{2}}(y_{jh} - t^*_{jh})+ \sum\limits_{h\in \mathcal{G}_{2}} t^*_{jh}\text{ \ \ \ } \text{\ }j\in \mathcal{J}\setminus B_1,  \label{B10ab}
\end{equation*}%

Since  for $\mathbf{t}^*$  we have
\begin{equation*}
0\leq  \sum\limits_{h\in \mathcal{G}_{2}}(y_{jh} - t^*_{jh}) \leq \epsilon \text{ \ \ \ }\text{\ }j\in \mathcal{J},  \label{B10ac}
\end{equation*}
and $x_{j1}$ is integral for $\mathcal{J}\setminus B_1$ the $\mathbf{t}^*$ satisfies the (\ref{AB10af}).

To prove  (\ref{AB10affz}) we observe that by Lemma \ref{ac} each $j\in B_1$ is $a$-tight and thus
\begin{equation}
\sum\limits_{h\in \mathcal{G}_{2}}b_{jh}+a_{j1}- x_{j1}- \Delta (%
\mathcal{G}_{1})+r  =  \sum\limits_{h\in \mathcal{G}_{2}}(y_{jh} - t^*_{jh})+ \sum\limits_{h\in \mathcal{G}_{2}} t^*_{jh} \text{ \ \ \ } \text{\ }j\in B_1. \label{AB10abcd}
\end{equation}%
%By Lemma \ref{onecolumn} and definition of $\mathbf{t}^*$
By Theorem \ref{tightd} $j\in B_1$ is $d$-tight. Thus by Lemma \ref{onecolumn}  and definition of truncated solution we have
\begin{equation}
\epsilon=\sum\limits_{h\in \mathcal{G}_{2}}(y_{jh} - t^*_{jh}),
%\geq \epsilon - \lambda \text{ \ \ \ }\forall \text{\ }j\in B_1. 
\label{AB10abcdd}
\end{equation}
for $j\in B_1$.

Thus by (\ref{AB10abcd}) and  (\ref{AB10abcdd}) 
\begin{equation*}
\sum\limits_{h\in \mathcal{G}_{2}}b_{jh}+a_{j1} - \Delta (%
\mathcal{G}_{1})+\lfloor r  \rfloor -\lfloor x_{j1} \rfloor +\epsilon -\epsilon - \varepsilon_j=   \sum\limits_{h\in \mathcal{G}_{2}} t^*_{jh}\text{ \ \ \ }\text{\ }j\in B_1. \label{B10abcdd1}
\end{equation*}
Hence (\ref{AB10affz}) is met by the truncated solution $\mathbf{t}^*$ and $\mathbf{z}^*$.
Therefore the truncated solution  $\mathbf{t}^*$ and $\mathbf{z}^*$ is feasible for $\cal{F}$.

To prove the lower bound on the value of objective function
we observe that  by (\ref{AB10abcd}) and (\ref{AB10abcdd})
\begin{equation}
\sum\limits_{h\in \mathcal{G}_{2}}b_{jh}+a_{j1}- x_{j1}- \Delta (%
\mathcal{G}_{1})+\lfloor r  \rfloor +\epsilon -\epsilon=   \sum\limits_{h\in \mathcal{G}_{2}} t^*_{jh}\text{ \ \ \ } \text{\ }j\in B_1. \label{B10abcdd1}
\end{equation}
Summing up (\ref{B10abcdd1}) side by side over all $j\in B_1$ we get by (\ref{bm1}) for  $(\mathbf{y},\mathbf{x},r,w)$
\begin{equation*}
\sum\limits_{j\in B_1}( \sum\limits_{h\in \mathcal{G}_{2}}b_{jh}+a_{j1})- (r-c)- |B_1|(\Delta (%
\mathcal{G}_{1})-\lfloor r  \rfloor )=   \sum\limits_{j\in B_1} \sum\limits_{h\in \mathcal{G}_{2}} t^*_{jh}, \label{B10abcdd12}
\end{equation*}
where $c=\sum\limits_{j\in \mathcal{J}\setminus B_1} x_{j1}$ is integral by definition of $B_1$.
Thus
\begin{equation*}
\sum\limits_{j\in B_1} \sum\limits_{h\in \mathcal{G}_{2}}b_{jh} +\Delta (%
\mathcal{G}_{1}) -\lfloor r \rfloor- \sum\limits_{j\in \mathcal{J}\setminus B_1} (a_{j1}-x_{j1}) - |B_1|(\Delta (%
\mathcal{G}_{1})-\lfloor r  \rfloor )- \epsilon=   \sum\limits_{j\in B_1} \sum\limits_{h\in \mathcal{G}_{2}} t^*_{jh} \label{B10abcdde}
\end{equation*}
and
\begin{equation*}
\sum\limits_{j\in B_1} \sum\limits_{h\in \mathcal{G}_{2}}b_{jh} - \sum\limits_{j\in \mathcal{J}\setminus B_1} (a_{j1}-x_{j1}) - (|B_1-1)|(\Delta (%
\mathcal{G}_{1})-\lfloor r  \rfloor )- \epsilon =   \sum\limits_{j\in B_1} \sum\limits_{h\in \mathcal{G}_{2}} t^*_{jh} \label{B10abcdde}
\end{equation*}
as required.

\end{proof}

\begin{lemma} \label{lowerb10}
If $\sum_{j\in B_1}\varepsilon_j=\epsilon$, then
\begin{equation}\label{F}
F= \sum\limits_{j\in B_1} \sum\limits_{h\in \mathcal{G}_2}b_{jh} + \sum\limits_{j\in B_1} (a_{j1}-\lfloor x_{j1}\rfloor)- |B_1|(\Delta(\mathcal{G}_1)-\lfloor r \rfloor)
\end{equation}
and
\begin{equation} \label{F0}
\sum_{h\in \mathcal{G}_2}t_{jh}=\sum_{h\in \mathcal{G}_2}b_{jh} + a_{j1} - \lfloor x_{j1} \rfloor - \Delta(\mathcal{G}_1) + \lfloor r \rfloor \text{ \ \ \ } \text{\ }j\in B_1.
\end{equation}
\end{lemma}

\begin{proof} By  (\ref{B10abcdd1}) 
\begin{equation}
\sum\limits_{h\in \mathcal{G}_{2}}b_{jh}+a_{j1} - \Delta (%
\mathcal{G}_{1})+\lfloor r  \rfloor -\lfloor x_{j1} \rfloor - \varepsilon_j=   \sum\limits_{h\in \mathcal{G}_{2}} t^*_{jh}\text{ \ \ \ }\text{\ }j\in B_1, \label{Blue}
\end{equation}
summing up side by side for $j\in B_1$ and taking $\sum_{j\in B_1}\varepsilon_j=\epsilon$ we get
\begin{equation}
\sum_{j\in B_1}\sum\limits_{h\in \mathcal{G}_{2}}b_{jh}+\sum_{j\in B_1}(a_{j1} -\lfloor x_{j1} \rfloor)- |B_1|(\Delta (%
\mathcal{G}_{1})-\lfloor r  \rfloor ) - \epsilon=  \sum_{j\in B_1} \sum\limits_{h\in \mathcal{G}_{2}} t^*_{jh}, \label{Blue1}
\end{equation}
for the truncated solution $\mathbf{t}^*$ and $\mathbf{z}^*$, which by Lemma \ref{lowerb1} is feasible for $\cal{F}$. Let   $\mathbf{t}$ and $\mathbf{z}$ be an optimal solution for $\cal{F}$. Since all upper and lower bounds in $\cal{F}$ are integral, we may assume both $\mathbf{t}$ and $\mathbf{z}$
integral by the Integral Circulation Theorem, see \cite{L}. Thus by (\ref{Blue1})
\begin{equation}
\sum_{j\in B_1}\sum\limits_{h\in \mathcal{G}_{2}}b_{jh}+\sum_{j\in B_1}(a_{j1} -\lfloor x_{j1} \rfloor)- |B_1|(\Delta (%
\mathcal{G}_{1})-\lfloor r  \rfloor ) \leq  \sum_{j\in B_1} \sum\limits_{h\in \mathcal{G}_{2}} t_{jh}, \label{Blue2}
\end{equation}
and the upper bounds the in (\ref{AB10affz}) give
\begin{equation}
\sum_{j\in B_1}\sum\limits_{h\in \mathcal{G}_{2}}b_{jh}+\sum_{j\in B_1}(a_{j1} -\lfloor x_{j1} \rfloor)- |B_1|(\Delta (%
\mathcal{G}_{1})-\lfloor r  \rfloor ) \geq  \sum_{j\in B_1} \sum\limits_{h\in \mathcal{G}_{2}} t_{jh}. \label{Blue3}
\end{equation}
Hence by (\ref{Blue2}) and (\ref{Blue3}) we get
\begin{equation*}
\sum_{j\in B_1}\sum\limits_{h\in \mathcal{G}_{2}}b_{jh}+\sum_{j\in B_1}(a_{j1} -\lfloor x_{j1} \rfloor)- |B_1|(\Delta (%
\mathcal{G}_{1})-\lfloor r  \rfloor ) =  \sum_{j\in B_1} \sum\limits_{h\in \mathcal{G}_{2}} t_{jh}=F, \label{Blue4}
\end{equation*}
which proves (\ref{F}) in the lemma. Finally, in order to reach this optimal value all upper bounds in (\ref{AB10affz}) must be reached, which proves (\ref{F0}).
\end{proof}

\begin{theorem} \label{bb1}
For $\sum_{j\in B_1}\varepsilon_j=\epsilon$, an optimal solution to $\cal{F}$ can be extended to an integral feasible solution to $\ell p$ with  $lp=\lfloor r^* \rfloor<r$.
\end{theorem}

\begin{proof}
Let $\mathbf{t}$ and $\mathbf{z}$ be an optimal solution to $\cal{F}$. This solution exists since by Lemma \ref{lowerb1} there is a feasible solution to $F$. Since all upper and lower bounds in $F$ are integral, we may assume both $\mathbf{t}$ and $\mathbf{z}$
integral by the Integral Circulation Theorem, see \cite{L}. Thus by Lemma \ref{lowerb1} 
%\WKcomment{feasible solution for the circulation exists}
\begin{equation}
\sum\limits_{j\in B_1}  \sum\limits_{h\in \mathcal{G}_{2}} t_{jh}\geq \sum\limits_{j\in B_1} \sum\limits_{h\in \mathcal{G}_2}b_{jh} - \sum\limits_{j\in \mathcal{J}\setminus B_1} (a_{j1}-x_{j1})-(|B_1|-1)(\Delta(\mathcal{G}_1)-\lfloor r \rfloor). \label{bound}
\end{equation}
For the partial solution $((\mathbf{t},\mathbf{z}), r'=\lfloor r \rfloor, w'=\lfloor w \rfloor)$ we have:  (\ref{AB4afq}) implies (\ref{r1d}),  (\ref{AB5af}) and (\ref{AB5aaf}) imply (\ref{y4}), (\ref{AB2af}) implies (\ref{y2}), and (\ref{AB3af}) implies (\ref{y1}).
Let us now extend the solution $((\mathbf{t},\mathbf{z}), r'=\lfloor r \rfloor, w'=\lfloor w \rfloor)$ by setting $x^*_{j2}:=\lfloor x_{j2} \rfloor$, for $j\in B_2$ and $x^*_{j2}:= x_{j2}$ for $j\in \mathcal{J}\setminus B_2$. Since $\sum_{j\in B_2}\varepsilon_j=\epsilon $,  (\ref{bm2}) is met by this extension. Clearly (\ref{y5}) is also met for $\ell=2$. By (\ref{AB10aafk}) we have
\begin{equation*}
\sum\limits_{h\in \mathcal{G}_{1}}b_{jh}+a_{j2}- x_{j2}- \Delta (%
\mathcal{G}_{2})+ \lfloor r \rfloor \leq  \sum\limits_{h\in \mathcal{G}_{1}}z_{jh} \text{ \ \ \ } \text{\ }j\in \mathcal{J}\setminus B_2  \label{B10aa1}.
\end{equation*}%
Also, since $\lfloor r - x_{j2}\rfloor=\lfloor r \rfloor - \lfloor x_{j2}\rfloor$ for $j\in B_2$ we have 
\begin{equation*}
\sum\limits_{h\in \mathcal{G}_{1}}b_{jh}+a_{j2}- x^*_{j2}- \Delta (%
\mathcal{G}_{2})+ \lfloor r \rfloor \leq  \sum\limits_{h\in \mathcal{G}_{1}}z_{jh} \text{ \ \ \ }  \label{B10aa1}
\end{equation*}%
for $j\in B_2$ by (\ref{AB10aafk}),  and thus (\ref{r1a}) is met for the extended solution  $((\mathbf{t},\mathbf{z}), r'=\lfloor r \rfloor, w'=\lfloor w \rfloor)$, and $x^*_{j2}$ for $j\in \mathcal{J}$. 

We now extend this solution further by setting
\begin{equation}
x^*_{j1}:=\sum\limits_{h\in \mathcal{G}_{2}}b_{jh}+a_{j1}- \Delta (%
\mathcal{G}_{1})+ \lfloor r \rfloor  -  \sum\limits_{h\in \mathcal{G}_{2}}t_{jh} \label{B10aa3}
\end{equation}%
for   $j\in B_1$ and $x^*_{j1}:=x_{j1}$ for $j\in \mathcal{J}\setminus B_1$.
% and $x_{j1}> \epsilon$,
%and $x^*_{j1}=x_{j1}$ for $j\in \mathcal{J}\setminus B_1$. 
To prove that (\ref{r1c}) is met for the extended solution  $((\mathbf{t},\mathbf{z}), r'=\lfloor r \rfloor, w'=\lfloor w \rfloor)$, and $x^*_{j2}$, $x^*_{j1}$ for $j\in \mathcal{J}$ we need to show that 
%By (\ref{B10a}) we have
\begin{equation}
\sum\limits_{h\in \mathcal{G}_{2}}b_{jh}+a_{j1} -x^*_{j1} - \Delta (%
\mathcal{G}_{1})+\lfloor r  \rfloor \leq  \sum\limits_{h\in \mathcal{G}_{2}}t_{jh}   \label{B10a12}
\end{equation}%
for each $j\in \mathcal{J}$. By the definition (\ref{B10aa3}) this holds for $j\in B_1$.
% and $x_{j1}> \epsilon$. 
% By (\ref{B10aa2}) this also holds for $j\in B_1$ and $x_{j1}\leq \epsilon$ since then $\lfloor r\rfloor = \lfloor r- x_{j1} \rfloor$.
For $j\in \mathcal{J}\setminus B_1$ we have $x_{j1}$ integral and  thus (\ref{B10a12}) holds since (\ref{AB10af}) holds. Thus (\ref{r1c}) is met for the extended solution  $((\mathbf{t},\mathbf{z}), r'=\lfloor r \rfloor, w'=\lfloor w \rfloor)$, and $x^*_{j2}$, $x^*_{j1}$ for $j\in \mathcal{J}$. 
%which implies $x^*_{j1} - \lfloor r \rfloor + \lfloor r - x_{j1} \rfloor \leq 0$ for $j\in B_1$. 
%Thus (\ref{B11}) is met f the extended solution. 
Moreover $a_{j1}\geq x^*_{j1}\geq 0$ for each $j$ and thus (\ref{y5}) holds for $\ell=1$ in  this extended solution. It suffices to prove this for $j\in B_1$.
 
Then, since $\lfloor r\rfloor \geq \lfloor r \rfloor - \lfloor x_{j1} \rfloor$,  $x^*_{j1}\geq 0$ by (\ref{B10aa3}) and the right hand side inequality of (\ref{AB10affz}). 
%By the RHS of (\ref{AB10affw}) $x^*_{j1}\geq 0$ for $j\in B_1$ and $x_{j1}<\epsilon$.  
Moreover, $a_{j1}\geq \lceil x_{j1}\rceil$. Thus by the left hand side inequality  of  (\ref{AB10affz})

\begin{equation*}
\sum\limits_{h\in \mathcal{G}_{2}}b_{jh}- \Delta (%
\mathcal{G}_{1})+\lfloor r \rfloor \leq  \sum\limits_{h\in \mathcal{G}_{2}}t_{jh}  \label{B10az}
\end{equation*}%
%For fractional $r-x_{j1}$, we have by (\ref{B10aa2})
and by (\ref{B10aa3})
\begin{equation*}
x^*_{j1}=\sum\limits_{h\in \mathcal{G}_{2}}b_{jh}- \Delta (%
\mathcal{G}_{1})+\lfloor r \rfloor -  \sum\limits_{h\in \mathcal{G}_{2}}t_{jh} +a_{j1}
\leq a_{j1}. \label{B10azz}
\end{equation*}
%For integral $r-x_{j1}$, $x^*_{j1}=0$ by (\ref{B10aff}). 
Therefore (\ref{y5}) holds for $\ell=1$ for $j\in B_1$.
For $j\in \mathcal{J}\setminus B_1$ the  (\ref{y5}) for $\ell=1$ in  the extended solution  $((\mathbf{t},\mathbf{z}), r'=\lfloor r \rfloor, w'=\lfloor w \rfloor)$, and $x^*_{j2}$, $x^*_{j1}$  follows from (\ref{y5}) for $\ell=1$ in the solution $(\mathbf{y},\mathbf{x},r,w)$.

By definition of the extended solution $((\mathbf{t},\mathbf{z}), r'=\lfloor r \rfloor, w'=\lfloor w \rfloor)$, and $x^*_{j2}$, $x^*_{j1}$ for $j\in \mathcal{J}$, and since by Theorem \ref{T4} there are no crossing jobs we have
\begin{equation} \label{f}
x^*_{j1}+x^*_{j2}\leq \lfloor r \rfloor
\end{equation}
for $j\in \mathcal{J}\setminus B_1$.  We now need to show this inequality for $j\in B_1$. For these jobs by the left hand side inequality of  (\ref{AB10affz}), and by (\ref{B10aa3})
% we have $x^*_{j1}=0$, by Theorem \ref{T4} we have integral $x_{j2}=x^*_{j2}$, and by (\ref{B6}) in the solution $(\mathbf{y},\mathbf{x},r,w)$ we have $x_{j2}\leq r$. Therefore (\ref{f}) holds for $j \in B_1$ and $x_{j1}\leq \epsilon$. For $j \in B_1$ and $x_{j1}>\epsilon$, by (\ref{B10aff}) and (\ref{B10aa2}) we get  
we get  $x^*_{j1} - \lfloor r \rfloor + \lfloor r \rfloor - \lceil x_{j1} \rceil \leq 0$.
Thus $x^*_{j1}\leq \lceil x_{j1} \rceil$ for each job $j\in B_1$. This unfortunately does not guarantee (\ref{f}) for  $j\in B_1$. However, we either have $\lceil x_{j1} \rceil + x_{j2}\leq \lfloor r \rfloor$ for each $j\in B_1$, in which case (\ref{f}) holds for  $j\in B_1$, or $\lceil x_{k1} \rceil + x_{k2}> \lfloor r \rfloor$ for some $k\in B_1$. The latter implies $\sum_{j\in B_1}\varepsilon_j=\epsilon$, which by Lemma \ref{lowerb10}, implies
\begin{equation*}  
 \sum\limits_{h\in \mathcal{G}_{2}}t_{jh}= \sum\limits_{h\in \mathcal{G}_{2}}b_{jh}+a_{j1}- \Delta (%
\mathcal{G}_{1})+ \lfloor r \rfloor - \lfloor x_{j1} \rfloor \text{ \ \ \ } \text{\ }j\in B_1 \label{CC1}
\end{equation*}
in the optimal solution $\mathbf{t}$ and $\mathbf{z}$ to $\cal{F}$. Thus by definition (\ref{B10aa3}), $x^*_{j1}=\lfloor x_{j1} \rfloor$ for $j\in B_1$. Since by Theorem \ref{T4} there are no crossing jobs the (\ref{f}) is satisfied. 
Hence it remain to prove that  if $\lceil x_{k1} \rceil + x_{k2}> \lfloor r \rfloor$ for some $k\in B_1$, then $\sum_{j\in B_1}\varepsilon_j=\epsilon$.
For contradiction assume $\lceil x_{k1} \rceil + x_{k2}> \lfloor r \rfloor$ for some $k\in B_1$ and $\sum_{j\in B_1}\varepsilon_j>\epsilon$.
If $x_{j1}x_{j2}=0$ for each $j\in \mathcal{J}$, then $x_{k2}=0$. Thus $\lceil x_{k1} \rceil > \lfloor r \rfloor$ which implies $\sum_{j\in B_1} \varepsilon_j=\epsilon$ and gives contradiction.
Otherwise, if $x_{i1}x_{i2}>0$ for some $i\in \mathcal{J}$, then by Theorem \ref{product} we have $B_1=\{i\}$ or $B_2=\{i\}$. If $B_1=\{i\}$, then $\sum_{j\in B_1} \varepsilon_j=\epsilon$ which gives contradiction.
Hence $B_2=\{i\}$ and $x_{j2}=0$ for each $j\in B_1$. Since by Theorem \ref{T4} there are no crossing jobs and $x_{i1}$ is integral and positive. Thus $x_{i1}\geq 1$, and $i\neq k$. By (\ref{bm1}) $\sum_j x_{j1}=\sum_{j\neq i} x_{j1} + x_{i1}=r$. Hence
$\sum_{j\neq i} x_{j1} \leq r - 1$ which gives $x_{k1} \leq r - 1$. Since $x_{k2}=0$, we get $x_{k1}+1+x_{k2}\leq r$. Thus $\lceil x_{k1} \rceil + x_{k2}\leq \lfloor r \rfloor$ which again gives contradiction. This proves that
if $\lceil x_{k1} \rceil + x_{k2}> \lfloor r \rfloor$ for some $k\in B_1$, then $\sum_{j\in B_1}\varepsilon_j=\epsilon$ as required.
Hence (\ref{y3}) holds for the extended solution  $((\mathbf{t},\mathbf{z}), r'=\lfloor r \rfloor, w'=\lfloor w \rfloor)$, and $x^*_{j2}$, $x^*_{j1}$.

Finally we need to prove that (\ref{bm1}) holds for an extended solution. 
By (\ref{B10aa3}) and (\ref{bound})
\begin{equation} \label{r1}
\sum\limits_{j}x^*_{j1}\leq \lfloor r \rfloor
\end{equation}
for the extended solution   $( (\mathbf{t},\mathbf{z},  \lfloor r \rfloor, \lfloor w \rfloor))$, and $x^*_{j2}$, $x^*_{j1}$ for $j\in \mathcal{J}$. This solution satisfies all constraints (\ref{y1})-(\ref{y4}) and  (\ref{bm2})-(\ref{r1c})  of $\ell p$. To complete the proof it suffices to modify the extension $x^*_{j1}$ for $j\in \mathcal{J}$ in order to ensure the equality in (\ref{r1}) to satisfy (\ref{bm1}), and to keep other constraint (\ref{y1})-(\ref{y4}) and  (\ref{bm2})-(\ref{r1c})  of $\ell p$ satisfied. 
%If $B_1=\{j\}$, then $x^*_{j1}<\lfloor x_{j1}\rfloor$ for the sharp inequality in (\ref{r1}). Then set $x^*_{j1}:=\lfloor x_{j1} \rfloor$ to satisfy all constraints (\ref{B2})-(\ref{B11}) in $LP$. 

 If $\sum\limits_{j}x^*_{j1}< \lfloor r \rfloor$, then take a  $j\in B_1$ with a positive $d_j=\min \{\lceil x_{j1} \rceil - x^*_{j1}, \lfloor r \rfloor - x^*_{j1}-x_{j2}\}$. Recall that by Theorem \ref{T4}, $x_{j2}$ is integral for each $j\in B_1$. Such $j$ exists. To prove this existence define $X=\{j\in B_1:\lceil x_{j1} \rceil=x^*_{j1}\}$ and $Y=\{j\in B_1:x^*_{j1}= \lfloor x_{j1} \rfloor\}$.  By definition (\ref{B10aa3}) and (\ref{AB10affz})
we have $B_1=X\cup Y$, and since
\begin{equation} \label{jeden}
\sum\limits_{j}x^*_{j1}< \lfloor r \rfloor < \sum\limits_{j}\lceil x_{j1} \rceil
\end{equation}
we have $Y\neq \emptyset$.
Suppose for a contradiction that for each job $j\in Y$ we have $\lfloor r \rfloor=x^*_{j1}+x_{j2}$. 
%Hence $x^*_{j1}=\lfloor x_{j1} \rfloor$ for each $j\in Y$. 
Thus we have
\begin{equation*} \label{x1}
\sum\limits_{j}x^*_{j1}= \sum\limits_{j\in \mathcal{J}\setminus B_1}x_{j1}+ \sum\limits_{j\in X}\lceil x_{j1}\rceil +\sum\limits_{j\in Y}\lfloor x_{j1} \rfloor< \lfloor r \rfloor.
\end{equation*}
Since for each job $j\in Y$ we have $\lfloor r \rfloor=\lfloor x_{j1}\rfloor +x_{j2}$, we obtain
\begin{equation*}
\sum\limits_{j\in \mathcal{J}\setminus B_1}x_{j1} + \sum\limits_{j\in X}\lceil x_{j1}\rceil + |Y|\lfloor  r \rfloor - \sum\limits_{j\in Y} x_{j2} < \lfloor r \rfloor,
\end{equation*}
and by (\ref{jeden}) the set $Y$ is not empty. Since $\sum\limits_{j\in Y} x_{j2}\leq \lfloor r \rfloor$ by (\ref{bm2}) we get
\begin{equation*}
\sum\limits_{j\in \mathcal{J}\setminus B_1}x_{j1} +\sum\limits_{j\in X}\lceil x_{j1}\rceil + |Y|\lfloor  r \rfloor < 2\lfloor r \rfloor,
\end{equation*}
and thus $|Y|\leq 1$, and since $Y$ is not empty we have $|Y|=1$. However
\begin{equation*}
\lfloor r \rfloor = \lfloor \sum\limits_{j}x_{j1} \rfloor= \sum\limits_{j}\lfloor x_{j1} \rfloor + \lfloor \sum\limits_{j\in B_1}\epsilon_{j} \rfloor,
\end{equation*}
where
\begin{equation*}
 \lfloor \sum\limits_{j\in B_1}\varepsilon_{j} \rfloor \leq |B_1|-1.
\end{equation*}
Thus
\begin{equation*}
\lfloor r \rfloor =\lfloor \sum\limits_{j}x_{j1} \rfloor\leq \sum\limits_{j\in \mathcal{J}\setminus B_1}x_{j1} + \sum\limits_{j\in B_1}\lfloor x_{j1} \rfloor+ |B_1|-1=\sum\limits_{j\in \mathcal{J}\setminus B_1}x_{j1}+\sum\limits_{j\in X}\lceil x_{j1} \rceil +\sum\limits_{j\in Y}\lfloor x_{j1} \rfloor
\end{equation*}
since $|Y|=1$
%\begin{equation}
%\lfloor r \rfloor =\lfloor \sum\limits_{j}x_{j1} \rfloor\leq \sum\limits_{j\in \mathcal{J}\setminus B_1}x_{j1} + \sum\limits_{j\in X}\lceil x_{j1} \rceil +\sum\limits_{j\in Y}\lfloor x_{j1} \rfloor
%\end{equation}
which contradicts (\ref{x1}) and proves that $j\in Y$ with $d_j=1$ exists. Set $d:=\min \{\min_{j, d_j>0}\{d_j\}, \lfloor r \rfloor - \sum\limits_{j}x^*_{j1}\}=1$. Then,
set $x^*_{j1}:=x^*_{j1} +1$ for some  $j\in Y$ with  $d_j=1$. We have $x^*_{j1}\leq \min \{\lceil x_{j1} \rceil, \lfloor r \rfloor - x_{j2} \}$ and $\sum\limits_{j}x^*_{j1}\leq \lfloor r \rfloor$ for the new extended solution, which ensures that  all constraints (\ref{y1})-(\ref{y4}) and  (\ref{bm2})-(\ref{r1c})  of $\ell p$ are met in the new extended solution. Since $d=1$ the $\sum\limits_{j}x^*_{j1}$ gets closer to but does not exceed $\lfloor r \rfloor$. Therefore by (\ref{jeden}) we finally reach an extended solution $\mathbf{t}$, $\mathbf{z}$, and $x^*_{j2}$, $x^*_{j1}$ for $j\in \mathcal{J}$ that meets all (\ref{y1})-(\ref{r1c})  of $\ell p$. The solution is integral with $w'=\lfloor w \rfloor$, and  $r'=\lfloor r^* \rfloor$ which proves the lemma.
\end{proof}

\section{The Projection} \label{S11}

Consider the following system $S$ that defines the set of feasible solutions to the $LP$-relaxation of $\cal{ILP}$,

%\begin{equation}w -r =\left \lceil w^{ \ast } -r^{ \ast }\right \rceil  \label{XB20}
%\end{equation}
\begin{equation}\sum _{j}b_{j h} -( \Delta (\mathcal{G}_{2}) -r) \leq \sum _{j}y_{j h} \leq w\;\;\text{\ \ \ }\;\; \;\;\text{}\;\;h \in \mathcal{G}_{1} \label{XXy1}
\end{equation}
\begin{equation}\sum _{j}b_{j h} -( \Delta (\mathcal{G}_{1}) -r) \leq \sum _{j}y_{j h} \leq w\;\;\text{\ \ }\;\; \;\;\text{\ }\;\;h \in \mathcal{G}_{2} \label{XXy2}
\end{equation}
\begin{equation}\sum _{h}y_{j h} \leq w\;\;\text{\ }\;\; \;\;\text{\ }\;\;j\in \mathcal{J} \label{XXr1d}
\end{equation}
\begin{equation}0 \leq y_{j h} \leq b_{j h} \;\;\text{\ }\;\; \;\;\text{}\;\;h \in \mathcal{M}\;\; \;\;\text{\ }\;\;j \in \mathcal{J}\label{XXy4}
\end{equation}
\begin{equation}\sum _{j}x_{j 1} = r \label{XXbm1}
\end{equation}
\begin{equation}\sum _{j}x_{j 2} = r \label{XXbm2}
\end{equation}
\begin{equation}x_{j 1} +x_{j 2} \leq r\;\;\text{\ }\;\; \;\;\text{\ \ }\;\;j\in \mathcal{J} \label{XXy3}
\end{equation}
\begin{equation}0 \leq x_{j \ell } \leq a_{j \ell } \;\; \;\;\text{\ }\;\;j \in \mathcal{J} \;\; \;\;\text{\ }\;\;\ell=1,2\label{XXy5}
\end{equation}
\begin{equation}\sum _{h \in \mathcal{G}_{1}}(b_{j h} -y_{j h}) +a_{j 2} -x_{j 2} \leq  \Delta (\mathcal{G}_{2}) -r\;\;\text{\ \ }\;\; \;\;\text{\ \ }\;\;j\in \mathcal{J} \label{XXr1a}
\end{equation}
\begin{equation}\sum _{h \in \mathcal{G}_{2}}(b_{j h} -y_{j h}) +a_{j 1} -x_{j 1} \leq  \Delta (\mathcal{G}_{1}) -r\;\;\text{\ \ }\;\; \;\;\text{\ \ }\;\;j\in \mathcal{J} \label{XXr1c}
\end{equation}

Now consider the system $S_r$ obtained from $S$ by dropping (\ref{XXbm1}) and (\ref{XXbm2})  and adding the constraints (\ref{YY}), (\ref{YYY}), and (\ref{Y4}). 
We use $\alpha_{j1}=\sum_{h\in \mathcal{G}_2}(b_{jh}-y_{jh})+a_{j1}-\Delta(\mathcal{G}_1)$
and $\alpha_{j2}=\sum_{h\in \mathcal{G}_1}(b_{jh}-y_{jh})+a_{j2}-\Delta(\mathcal{G}_2)$ for $j\in \mathcal{J}$ for convenience. 
%\begin{equation}w -r =\left \lceil w^{ \ast } -r^{ \ast }\right \rceil  \label{YB20}
%\end{equation}
\begin{equation}\sum _{j}b_{j h} -( \Delta (\mathcal{G}_{2}) -r) \leq \sum _{j}y_{j h} \leq w\;\;\text{\ \ \ }\;\; \;\;\text{}\;\;h \in \mathcal{G}_{1} \label{Yy1}
\end{equation}
\begin{equation}\sum _{j}b_{j h} -( \Delta (\mathcal{G}_{1}) -r) \leq \sum _{j}y_{j h} \leq w\;\;\text{\ \ }\;\; \;\;\text{\ }\;\;h \in \mathcal{G}_{2} \label{Yy2}
\end{equation}
\begin{equation}\sum _{h}y_{j h} \leq w\;\;\text{\ }\;\; \;\;\text{\ }\;\;j \in \mathcal{J} \label{Yr1d}
\end{equation}
\begin{equation}0 \leq y_{j h} \leq b_{j h} \;\;\text{\ }\;\; \;\;\text{}\;\;h \in \mathcal{M}\;\; \;\;\text{\ }\;\;j \in \mathcal{J}\label{Yy4}
\end{equation}
%\begin{equation}\sum _{j}x_{j 1} =r \label{bm1}
%\end{equation}
%\begin{equation}\sum _{j}x_{j 2} =r \label{bm2}
%\end{equation}
\begin{equation}x_{j 1} +x_{j 2} \leq r\;\;\text{\ }\;\; \;\;\text{\ \ }\;\;j \in \mathcal{J} \label{Yy3}
\end{equation}
\begin{equation}0 \leq x_{j \ell } \leq a_{j \ell }\;\; \;\;\text{\ }\;\;j \in \mathcal{J} \;\; \;\;\text{\ }\;\;\ell=1,2 \label{Yy5}
\end{equation}
\begin{equation}\sum _{h \in \mathcal{G}_{1}}(b_{j h} -y_{j h}) +a_{j 2} -x_{j 2} \leq  \Delta (\mathcal{G}_{2}) -r\;\;\text{\ \ }\;\;  \;\;\text{\ \ }\;\;j \in \mathcal{J}\label{Yr1a}
\end{equation}
\begin{equation}\sum _{h \in \mathcal{G}_{2}}(b_{j h} -y_{j h}) +a_{j 1} -x_{j 1} \leq  \Delta (\mathcal{G}_{1}) -r\;\;\text{\ \ }\;\; \;\;\text{\ \ }\;\;j \in \mathcal{J}\label{Yr1c}
\end{equation}
\begin{equation} \label{YY}
\sum_j \alpha_{j1} + (n-1)r\leq 0 \;\;\text{\ \ }\;\; \;\;\text{\ \ }\;\;j \in \mathcal{J}
\end{equation}
\begin{equation} \label{YYY}
\sum_j \alpha_{j2} + (n-1)r\leq 0 \;\;\text{\ \ }\;\; \;\;\text{\ \ }\;\;j \in \mathcal{J}
\end{equation}
\begin{equation} \label{Y4}
0\leq r \leq \min\{\Delta(\mathcal{G}_1),\Delta(\mathcal{G}_2)\}
\end{equation}

Finally consider the following projection on $\bold{y,} w, r$.

\begin{lemma} \label{LL3}
Let $\cal{P}$ be the polyhedron that consists of feasible solutions to $S_r$. Then the projection of $\cal{P}$ on $\bold{y,} w, r$, denoted by $\cal{Q}$,  is the set of solutions to the following system of inequalities $Q$:

%\begin{equation}
%w-r=\lceil w^* - r^* \rceil
%\end{equation}

\begin{equation}\sum _{j}b_{j h} -( \Delta (\mathcal{G}_{2}) -r) \leq \sum _{j}y_{j h} \leq w\;\;\text{\ \ }\;\;  \;\;\text{\ }\;\;h \in \mathcal{G}_{1} \label{Pr00}
\end{equation}
\begin{equation}\sum _{j}b_{j h} -( \Delta (\mathcal{G}_{1}) -r) \leq \sum _{j}y_{j h} \leq w\;\;\text{\ \ \ }\;\; \;\;\text{}\;\;h \in \mathcal{G}_{2} \label{Pr01}
\end{equation}
\begin{equation}\sum _{h}y_{j h} \leq w\;\;\text{\ }\;\; \;\;\text{\ }\;\;j \in \mathcal{J} \label{Pr02}
\end{equation}
\begin{equation}0 \leq y_{j h} \leq b_{j h} \;\;\text{\ }\;\; \;\;\text{}\;\;h \in \mathcal{M}\;\; \;\;\text{\ }\;\;j \in \mathcal{J}\label{Pr03}
\end{equation}

\begin{equation}\label{Pr1}
\sum_{h\in \mathcal{G}_2}(b_{jh}-y_{jh}) + a_{j1} - \Delta(\mathcal{G}_1)\leq 0  \; \; \; \;  \;j\in \mathcal{J}
\end{equation}

\begin{equation}\label{Pr2}
\sum_{h\in \mathcal{G}_1}(b_{jh}-y_{jh}) + a_{j2} - \Delta(\mathcal{G}_2)\leq 0 \; \; \; \;  \;j\in \mathcal{J}
\end{equation}

\begin{equation}\label{Pr3}
\sum_{h\in \mathcal{G}_2}(b_{jh}-y_{jh}) + r - \Delta(\mathcal{G}_1)\leq 0  \; \; \; \;  \;j\in \mathcal{J}
\end{equation}

\begin{equation}\label{Pr4}
\sum_{h\in \mathcal{G}_1}(b_{jh}-y_{jh}) + r - \Delta(\mathcal{G}_2)\leq 0  \; \; \; \; \;j\in \mathcal{J}
\end{equation}

\begin{equation}\label{Pr7}
\sum_h (b_{jh}-y_{jh}) + a_{j1} + a_{j2} - \Delta (\mathcal{G}_1) -\Delta (\mathcal{G}_2) + r  \leq 0 \; \; \; \; \;j\in \mathcal{J}
\end{equation}

\begin{equation}\label{PR10}
\sum_j \alpha_{j1} + (n-1)r\leq 0 \;\;\text{\ \ }\;\; \;\;\text{\ \ }\;\;j \in \mathcal{J}
\end{equation}

\begin{equation}\label{PR11}
\sum_j \alpha_{j2} + (n-1)r\leq 0 \;\;\text{\ \ }\;\; \;\;\text{\ \ }\;\;j \in \mathcal{J}
\end{equation}

\begin{equation}\label{Pr10}
0\leq r\leq \min\{\Delta(\mathcal{G}_1),\Delta(\mathcal{G}_2)\}
\end{equation}

\end{lemma}

\begin{proof}
By Fourier-Motzkin elimination, see \cite{Sch}, of variables $x_{j\ell}$ from the system $S_r$.
\end{proof}

We summarize the results of this section in the following theorem and lemma.

\begin{theorem} \label{T11}
Let   $(\bold{y}, r, w)$ be feasible for $Q$. There exists  $\bold{x}$  such that the solution  $(\bold{y}$, $\bold{x}, w, r)$ is feasible for $S$.
\end{theorem}

\begin{proof}
Let $s=(\bold{y}, r, w)$ be a feasible solution for $Q$. By Lemma  \ref{LL3} there exist $\bold{x}=(x_{j\ell})$, where $j\in \mathcal{J}$ and $\ell=1,2$, such that $s=(\bold{y}$, $\bold{x}, w, r)$ is feasible for $S_r$. Let $X$ be the set of all such $\bold{x}$.
Take $\bold{x}\in X$ with minimum distance $d=|r-\sum_j x_{j1}| + |r-\sum_j x_{j2}|$. We show that $d=0$ for $\bold{x}$. Suppose that $r< \sum_j x_{j1}$ or $r< \sum_ j x_{j2}$. 
Let $r< \sum_j x_{j1}$. If there is $k$ such that $\alpha_{k1} + r < x_{k1}$, then set $x_{k1}:=x_{k1}-\lambda$
where $\lambda=\min\{ x_{k1} -(\alpha_k + r ), \sum_j x_{j1} - r\}$. The new solution is in $X$ and reduces $d$ which gives a contradiction.
Thus we have $\alpha_{j1} + r = x_{j1}$ for each $j$. Therefore $\sum_j \alpha_{j1} + nr=\sum_j x_{j1}\leq r$ by the constraint (\ref{PR10}) which contradicts this case assumption.
%Let $P$ be the set of all $j$ such that $x_{j1}>0$. The set is not empty since otherwise $x_j= 0$ for each $j$ thus
%$r<0$ which contradicts (\ref{Pr10}). Let $Q$ be the set of all $j$ in $\mathcal{J}\setminus P$ such that $x_{j2}>0$. If $Q=\emptyset$, then set $r:=r-\lambda$ and $x_{j1}:=x_{j1}-\lambda$ for $j\in P$  where $\lambda=\min_{j\in P}\{x_{j1}\}$. 
%If $Q\neq \emptyset$, then set $r:=r-\lambda$ and $x_{j1}:=x_{j1}-\lambda$ for $j\in P$and $x_{j2}:=x_{j2}-\lambda$ for $j\in Q$  where $\lambda=\min\{\min_{j\in P}\{x_{j1}\}, \min_{j\in Q}\{x_{j2}\}\}$. 
%In either case the new solution is feasible but has smaller $\sum_j (x_{j1}+x_{j2})$ which gives a contradiction. 
The proof for $r< \sum_ j x_{j2}$ is similar.
Therefore we have $r\geq \sum_j x_{j1}$ and $r\geq \sum_ j x_{j2}$ for the $\bold{x}$.  Suppose that $r> \sum_j x_{j1}$ or $r> \sum_j x_{j2}$. If there is $k$ such that $x_{k1}+x_{k2}<r$ and ($x_{k1}<a_{k1}$ or $x_{k2}<a_{k2}$), then set $x_{k1}+\lambda$,
where $\lambda=\min\{r-(x_{k1}+x_{k2}), a_{k1} - x_{k1}, d\}$ provided $x_{k1}<a_{k1}$. Otherwise, if $x_{k1}=a_{k1}$ and $x_{k2}<a_{k2}$, set $x_{k2}+\lambda$,
where $\lambda=\min\{r-(x_{k1}+x_{k2}), a_{k2} - x_{k2}, d\}$. The new solution is  in $X$ but has smaller $d$ which gives a contradiction. Thus we have
$x_{j1}+x_{j2}=r$ or ($x_{j1}=a_{j1}$ and $x_{j2}=a_{j2}$) for each $j$. We have at least one $j$ with $x_{j1}+x_{j2}=r$. Otherwise, $r>\min\{\Delta(\mathcal{G}_1),\Delta(\mathcal{G}_2)\}$ which contradicts 
(\ref{Pr10}).
 On the other hand,
we can have at most one $j$ with $x_{j1}+x_{j2}=r$. Otherwise $\sum_j (x_{j1}+x_{j2})\geq 2r$ and since $r\geq \sum_j x_{j1}$ and $r\geq \sum_ j x_{j2}$ for the $\bold{x}$  we get $r= \sum_j x_{j1}$ and $r= \sum_j x_{j2}$ which contradicts the assumption. Therefore there is exactly one $j$ such that $x_{j1}+x_{j2}=r$, and $x_{k1}=a_{k1}$, and $x_{k2}=a_{k2}$ for  $k\in \mathcal{J}\setminus \{j\}$. Hence $\Delta(\mathcal{G}_1)-a_{j1}+x_{j1}< r$ or $\Delta(\mathcal{G}_2)-a_{j2}+x_{j2}< r$. Since 
$\Delta(\mathcal{G}_1)-a_{j1}+x_{j1}\leq  r$ and $\Delta(\mathcal{G}_2)-a_{j2}+x_{j2}\leq r$, we have $\Delta(\mathcal{G}_1) + \Delta(\mathcal{G}_2)-a_{j2}+x_{j2}-a_{j1}+x_{j1}< 2r$. Hence $\Delta(\mathcal{G}_1) + \Delta(\mathcal{G}_2)-a_{j2}-a_{j1}< r$
since $x_{j1}+x_{j2}=r$. However by (\ref{Pr7}) and (\ref{Pr03}) we have $a_{j1}+a_{j2}+r \leq \Delta(\mathcal{G}_1) + \Delta(\mathcal{G}_2)$ which gives a contradiction.
% If $x_{j1}>0$ and $x_{k2}>0$ for some $k\in \mathcal{J}$, then
%We have $r>0$, Otherwise $\sum_j (x_{j1}+x_{j2})\geq 2r$ which contradicts the case assumption. 
%set $r:= r -\lambda$, $x_{j1}:=x_{j1} - \lambda$,  $x_{l1}:=x_{l1}$ for $k\in \mathcal{J}\setminus\{j\}$, and  $x_{k2}:=x_{k2} - \lambda$,  $x_{l2}:=x_{l2}$ for $l\in \mathcal{J}\setminus\{k\}$, where $\lambda=\min\{ x_{j1}, \min_{j\neq k} \{r-(x_{k1}+x_{k2}), \min_k\{x_{k2}\}\}\}$.
%If $x_{j2}>0$ and $x_{k1}>0$ for some $k\in \mathcal{J}$, then set $r:= r -\lambda$ and $x_{j2}:=x_{j2} - \lambda$, ,  $x_{l2}:=x_{l2}$ for $k\in \mathcal{J}\setminus\{j\}$, and  $x_{k1}:=x_{k1} - \lambda$,  $x_{l1}:=x_{l1}$ for $l\in \mathcal{J}\setminus\{k\}$, where $\lambda=\min\{x_{j2}, \min_{j\neq k} \{r-(x_{k1}+x_{k2})\}, \min_k\{x_{k1}\}\}$.  if $x_{j2}>0$ and $x_{k1}>0$ for some $k$. In either case we get a feasible solution with the same $d$ but
%smaller $r$ which leads to contradiction.  Finally, if $\sum_j x_{j1}=0$, then $x_{k1}+x_{k2}=0 $ for $k\in \mathcal{J}\setminus\{j\}$ and $x_{j2}=r;$  set $r:= 0$ and $x_{j2}:=0$. 
%Similarly, if $\sum_j x_{j2}=0$, then $x_{k1}+x_{k2}=0 $ for $k\in \mathcal{J}\setminus\{j\}$ and $x_{j1}=r;$  set $r:= 0$ and $x_{j1}:=0$. In either case we get a feasible solution with $d=0$ which leads to contradiction.
Thus we have $d=0$ and the solution is feasible for $S$.

%The theorem follows from Lemmas \ref{LL1}, \ref{LL2}, and \ref{LL3}.

\end{proof}

We have the following lemma

\begin{lemma}\label{LL4}
If $(\bold{y}$, $\bold{x}, r, w)$ is feasible for $S$, then $(\bold{y}, r, w)$ is feasible for $Q$.
\end{lemma}

\begin{proof}
If $(\bold{y}$, $\bold{x}, r, w)$ is feasible for $S$, then it is also feasible for $S_{r}$. Observe that (\ref{XXbm1}), (\ref{XXbm2}), and  (\ref{XXy5}) in $S$ imply (\ref{Y4}) in $S_r$, 
the (\ref{XXr1a}) in $S$ implies (\ref{YY}) in $S_r$, and the (\ref{XXr1c}) in $S$ implies (\ref{YYY}) in $S_r$. Finally, by Lemma \ref{LL3} the $(\bold{y}, r, w)$ is feasible for $Q$.

\end{proof}

The system $Q$ is  a network flow model with lower and upper bounds on the arcs for fixed $w$ and $r$.

\section{ Integral Optimal Solution to $\ell p$ for $\sum_{j\in B_i}\varepsilon_j>\epsilon$ for $i=1,2$} \label{S12}

%\subsection{A truncated solution feasible to the projection}

Consider $\mathbf{s}$ with $\sum_{j\in B_i}\varepsilon_j>\epsilon$ for $i=1,2$, by Lemma \ref{overlap} overlap of $B_1$ and of $B_2$ dos not occur simultaneously. Without loss of generality let us assume no overlap of $B_2$.  

Consider the set $Y_j$ of all columns  of type $\binom {%
\ast }{\ast ,j}$ in $d(\mathbf{y},w)$ for $j\in B_2$.
By Lemma \ref{l7},  $l(Y_j)=\beta_j=\lfloor \beta_j \rfloor + \varepsilon_j$.
Take an interval $Y \subseteq \bigcup_{j\in B_2} Y_j $ such that $l(Y)= \epsilon$. Such $Y$ exists since there is no overlap of $B_2$ and $\sum_{j\in B_2}\varepsilon_j>\epsilon$.
%, $l(Y\cap Y_j)\geq \epsilon_j$ for $j\in B_2$.
Let $Y_{jh}$ be the set of columns $I\in Y$ such that $(j,h)\in M_I$, set $\gamma_{jh}:=l(Z_{jh})$. Informally, $\gamma_{jh}$ is the amount of $j$ done on $h$ in the interval $Y$. We define a truncated solution as follows $z^*_{jh}:=y_{jh}-\gamma_{jh}$ for $h\in \mathcal{G}_1$, and  $t^*_{jh}:=y_{jh}-\gamma_{jh}$ for $h\in \mathcal{G}_2$, and $\lfloor r \rfloor$, $\lfloor w \rfloor$.
%By Theorem \ref{tightd} each $j\in B_2$ is $d$-tight
% and no overlap on $B_2$ 
 Thus
\begin{equation*}
%\sum\limits_{h\in \mathcal{G}_{2}}b_{jh}+a_{j1}- \Delta (%
%\mathcal{G}_{1})+\lfloor r - x_{j1}\rfloor \leq  
 \sum\limits_{h\in \mathcal{G}_{1}}\gamma_{jh}+  \sum\limits_{h\in \mathcal{G}_{2}}\gamma_{jh}\leq \epsilon \text{ \ \ \ } \text{\ }j\in \mathcal{J}  \label{DX}
\end{equation*}%
%and 
%\begin{equation} \label{CX1}
%\sum\limits_{h\in \mathcal{G}_{2}}\gamma_{jh}=\eta_j\geq \epsilon_j \text{ \ \ \ } \text{\ }j\in  B_2
%\end{equation}

\begin{theorem} \label{T10}
For $\sum_{j\in B_i}\varepsilon_j>\epsilon$ for $i=1,2$, there is a feasible integral solution to  $\ell p$ with $lp=\lfloor r^* \rfloor<r$.
\end{theorem}

\begin{proof} The constraints (\ref{Pr1}) and (\ref{Pr2}): For $\mathbf{s}$ we have

%\begin{equation} \label{FE1}
%\sum_{h \in \mathcal{G}_{1}} (b_{jh}-y_{jh}) + a_{j2} - x_{j2} \leq \Delta(\mathcal{G}_{2}) - r  \; \; \; \; \forall \;j\in \mathcal{J}
%\end{equation}

\begin{equation*}\label{FE2}
\sum_{h \in \mathcal{G}_{1}} b_{jh} + a_{j2}  - \Delta(\mathcal{G}_{2}) + r -  x_{j2}\leq    \sum_{h \in \mathcal{G}_{1}} y_{jh} \; \; \; \;  \;j\in \mathcal{J}
\end{equation*}

\begin{equation*}\label{FE20}
\sum_{h \in \mathcal{G}_{2}} b_{jh} + a_{j1}  - \Delta(\mathcal{G}_{1}) + r -  x_{j1}\leq    \sum_{h \in \mathcal{G}_{2}} y_{jh} \; \; \; \; \;j\in \mathcal{J}.
\end{equation*}

If  $r-x_{j1}\geq \epsilon$ and $r-x_{j2}\geq \epsilon$ for each $j\in \mathcal{J}$, then $\sum_{h \in \mathcal{G}_{1}} y_{jh} - (r-x_{j2})\leq \sum_{h \in \mathcal{G}_{1}} t_{jh}$ and $\sum_{h \in \mathcal{G}_{2}} y_{jh} - (r-x_{j1})\leq \sum_{h \in \mathcal{G}_{2}} t_{jh}$ for each $j$. Hence (\ref{Pr1}) and  (\ref{Pr2}) hold for the truncated solution. Otherwise, if $r-x_{j1}< \epsilon$ or $r-x_{j2}< \epsilon$ for some $j\in \mathcal{J}$,
then  $\lfloor r \rfloor\leq x_{j1}$ or $\lfloor r \rfloor\leq x_{j2}$ for some $j$. This implies $\sum_{j\in B_1}\varepsilon_j=\epsilon$ or $\sum_{j\in B_2}\varepsilon_j=\epsilon$  which contradicts the theorem's assumption.

The constraints (\ref{Pr3}) and (\ref{Pr4}): For $\mathbf{s}$ we have

\begin{equation*}\label{FE21}
\sum_{h \in \mathcal{G}_{2}} b_{jh} + r  - \Delta(\mathcal{G}_{1}) + a_{j1} -  x_{j1}\leq    \sum_{h \in \mathcal{G}_{2}} y_{jh} \; \; \; \; \;j\in \mathcal{J}
\end{equation*}

\begin{equation*}\label{FE22}
\sum_{h \in \mathcal{G}_{1}} b_{jh} + r  - \Delta(\mathcal{G}_{2}) + a_{j2} -  x_{j2}\leq    \sum_{h \in \mathcal{G}_{1}} y_{jh} \; \; \; \; \;j\in \mathcal{J}.
\end{equation*}

By constraint (\ref{y5}) and definition of the truncated solution 

\begin{equation*}\label{FE20}
\sum_{h \in \mathcal{G}_{1}} b_{jh} + \lfloor r \rfloor - \Delta(\mathcal{G}_{2}) \leq    \sum_{h \in \mathcal{G}_{1}} y_{jh} -\epsilon \leq \sum_{h \in \mathcal{G}_{1}} t_{jh}\; \; \; \; \;j\in \mathcal{J}
\end{equation*}

Hence (\ref{Pr3}) and (\ref{Pr4}) hold.

The constraints (\ref{Pr7}): For $\mathbf{s}$, by (\ref{r1a}) and (\ref{r1c}) we have

\begin{equation*}\label{FE9}
\sum_{h \in \mathcal{G}_{1}} (b_{jh} -  y_{jh})+ a_{j2}  - x_{j2} \leq \Delta(\mathcal{G}_{2}) - r 
\end{equation*}

and

\begin{equation*}\label{FE10}
\sum_{h \in \mathcal{G}_{2}} (b_{jh} -  y_{jh})+ a_{j1}  - x_{j1} \leq \Delta(\mathcal{G}_{1}) - r ,
\end{equation*}

by summing up the two side by side we get

\begin{equation*}\label{FE11}
\sum_{h} (b_{jh} -  y_{jh})+ a_{j1}  + a_{j2} - x_{j1} -x_{j2} \leq \Delta(\mathcal{G}_{1})  + \Delta(\mathcal{G}_{2}) - 2r 
\end{equation*}

or

\begin{equation*}\label{FE12}
\sum_{h} b_{jh} + a_{j1}  + a_{j2}  - \Delta(\mathcal{G}_{1})  - \Delta(\mathcal{G}_{2}) +\lfloor r \rfloor \leq  \sum_{h}  y_{jh} - r + x_{j1} +x_{j2} -\epsilon.
\end{equation*}

Since $ \sum_{h}  y_{jh}  -\epsilon\leq \sum_{h}  t_{jh}$, we have

\begin{equation*}\label{FE13}
 \sum_{h}  y_{jh} - r + x_{j1} +x_{j2} -\epsilon\leq  \sum_{h\in \mathcal{G}_1}  t_{jh} - r + x_{j1} +x_{j2}
\end{equation*}

But $ - r + x_{j1} +x_{j2}\leq 0$  by the constraint  (\ref{y3}) and thus we get

\begin{equation*}\label{FE14}
\sum_{h} b_{jh} + a_{j1}  + a_{j2}  - \Delta(\mathcal{G}_{1})  - \Delta(\mathcal{G}_{2}) +\lfloor r \rfloor \leq  \sum_{h}  t_{jh}  
\end{equation*}

which proves that (\ref{Pr7}) holds. The constraints (\ref{Pr00})-(\ref{Pr03}) are satisfied by definition of truncated solution. Finally, since
$|\mathcal{G}_1|\leq n-1$ and $|\mathcal{G}_2|\leq n-1$ we have the constraints (\ref{PR10}) and (\ref{PR11}) satisfied by the truncated
solution. Observe that $|\mathcal{G}_1|>n-1$ or $|\mathcal{G}_2|>n-1$ implies $|\mathcal{G}_1|=0$ or $|\mathcal{G}_2|=0$ since $|\mathcal{G}_1|+|\mathcal{G}_2|\leq n$. This however contradicts
the saturation.

Therefore the truncated solution $(y^*=(z^*,t^*),\lfloor r \rfloor, \lfloor w \rfloor)$ is feasible for $Q$, and by Theorem \ref{T11} there exists $\bold{x}^*$  such that  $(y^*=(z^*,t^*), \bold{x}^*, \lfloor r \rfloor, \lfloor w \rfloor)$ is feasible for $S$. Moreover $\lfloor r^* \rfloor \leq \lfloor r \rfloor$, and
$\lfloor w \rfloor - \lfloor r \rfloor=\lceil w^* - r^* \rceil$ since  $\mathbf{s}=(\mathbf{y} ,\mathbf{x}, r, w)$  is feasible for $\ell p$. Thus the solution $(y^*=(z^*,t^*), \bold{x}^*, \lfloor r \rfloor, \lfloor w \rfloor)$ is feasible for $\ell p$ and $lp=\lfloor r \rfloor$. For a feasible solution to $Q$ with integral $\lfloor w \rfloor$ and $\lfloor r \rfloor$
all lower and upper bounds in the network $Q$ are integral
%By (\ref{B20}), if $r$ is integer, then $w$ is integer. Using the projective method of \cite{KDWK}
thus we can  find in polynomial time an integral $\mathbf{y}^*$. Finally for given integral and fixed
$\lfloor r \rfloor,\lfloor w\rfloor$ and $\mathbf{y}^ *$ the $S$ becomes a network flow model with integral lower and upper bounds on the flows. Thus we can find in polynomial time an integral $\mathbf{x}^*$ such that the integer solution $(\mathbf{y}^* ,\mathbf{x}^* ,\lfloor r \rfloor, \lfloor w \rfloor)$ is feasible for $l p$ and $lp=\lfloor r \rfloor$. 
\end{proof}

\section{The proof of the conjecture} \label{S13}

We are now ready to prove Theorem \ref{integer} which proves the conjecture.

\begin{proof}
For contradiction suppose the optimal value for $\ell p$ is fractional, $lp=r=\lfloor r^* \rfloor + \epsilon$, where $\epsilon>0$. By Theorem \ref{bb1} there is a feasible integral solution to $\ell p$ with $lp=\lfloor r^* \rfloor$  for
$\sum_{j\in B_1}\varepsilon_j=\epsilon$ or $\sum_{j\in B_2}\varepsilon_j=\epsilon$. By Theorem \ref{T10} there is a feasible integral solution to $\ell p$ with $lp=\lfloor r^* \rfloor$  for
$\sum_{j\in B_1}\varepsilon_j>\epsilon$ and $\sum_{j\in B_2}\varepsilon_j>\epsilon$. 
Thus there is a feasible integral solution 
%$(\mathbf{y} ,\mathbf{x} ,r ,w)$ 
for $\ell p$ with $\lfloor r^ * \rfloor<r$. Hence there is a feasible solution
to $\ell p$ which is smaller than optimal $r$  which gives contradiction and proves the first part of the theorem. Thus optimal $\mathbf{s}$ has both $r$ and $w$ integer. The $\mathbf{s}$ is feasible  for $S$ and thus it is feasible for $Q$ by Lemma \ref{LL4}.
For a feasible solution to $Q$ with integral  $w $ and $r $
all lower and upper bounds in the network $Q$ are integral
%By (\ref{B20}), if $r$ is integer, then $w$ is integer. Using the projective method of \cite{KDWK}
thus we can find in polynomial time an integral $\mathbf{y}$. Finally for given integral and fixed
$ r, w$ and $\mathbf{y}$ the $S$ becomes a network with integral lower and upper bounds on the flows. Thus we can find in polynomial time an integral $\mathbf{x}$ such that the integer solution $(\mathbf{y} ,\mathbf{x},r, w)$ is feasible for $l p$ and $lp=r$.
\end{proof}
% and finally given integral
%$r , w$ and $\mathbf{y}$ we can find in polynomial time an integral $\mathbf{x}$. Therefore, we can find in polynomial time an optimal solution $(\mathbf{y} ,\mathbf{x} ,r ,w)$ to $lp$ which is integral. This proves the second part of the theorem.

%This solution is feasible to $ILP$ and
%$\lfloor w \rfloor - \lfloor r \rfloor=\lceil w^{ \ast } -r^{ \ast }\rceil  =\lceil L P\rceil $. Finally by definition of $L P$-relaxation we have $L P \leq I L P$ for any feasible solution to $I L P$. Thus $(\mathbf{y} ,\mathbf{x} ,r ,w)$ is an optimal solution to $I L P$ obtained in polynomial time. 
% and $w= \lfloor w^* \rfloor$.
%The optimum for $LP$ equals $w^*-r^*$=$\lfloor w^* \rfloor - \lfloor r^* \rfloor +\varepsilon_w -\varepsilon_r$. We have $\varepsilon_w\leq \varepsilon_r$, thus $\lceil w^* -r^* \rceil= \lfloor w^* \rfloor -\lfloor r^*\rfloor$. 
% in both cases which proves that the optimal value for $LP$ is integral. This proves the conjecture.

\section*{Acknowledgements}
The author  is grateful to Dominic de Werra for his insightful  comments.


\begin{thebibliography}{99}
\bibitem {DWAD} de Werra, D., Asratian, A.S., and Durand, S. "Complexity of some special types
of timetabling problems" \textit{Journal of Scheduling}, 2002, \textbf{5}, 171-183 

\bibitem {B}
Berge, C. "Hypergraphs", Elsevier, 1989 

\bibitem {BM} Bondy, J.A. and Murty U.S.R "Graph
Theory with Applications", Elsevier, 1976 

\bibitem {DWKK} de Werra, D., Kis, T. and Kubiak,
W., ``Preemptive open shop scheduling with multiprocessors: polynomial cases and applications'', \textit{Journal of Scheduling}, 2008, \textbf{11},
75-83 

\bibitem {KDWK} Kis, T., de Werra, D. and Kubiak, W., "A projective algorithm for preemptive
open shop scheduling with two multiprocessor groups," \textit{Operations Research Letters, }2010, \textbf{38,} 129-132 

\bibitem {GK}
Gabow, H.N., and Kariv, O., "Algorithms for edge coloring bipartite graphs and multigraphs," \textit{SIAM Journal on Computing,} 1982, \textbf{11},
117-129. 

\bibitem {ADW} Asratian, A.S. and de Werra, D., "A generalized class-teacher model
for some timtabling problems," \textit{European Journal of Operational Research,} 2002, \textbf{143}, 531-542. 

\bibitem {GS}
Gonzalez, T. and Sahni, S., "Open shop scheduling to minimize finish time," \textit{Journal of the ACM, }1976, \textbf{23,} 665-679 

\bibitem {L}
Lawler, E., "Combinatorial Optimization. Networks and Matroids", Dover Publications, 2001 

\bibitem {LLV}
Lawler, E.L., Luby, M.G., and Vazirani, V.V., "Scheduling open shops with parallel machines" \textit{Operations Research Letters, } 1982, \textbf{1,}
161-164 

\bibitem{Sch} Schrijver, A., "Theory of linear and integer programming" John Wiley \& Sons, 1999.

%\bibitem {LeV} LeVeque, W.J., "Topics in Number Theory" Volumes I and II, Dover Publications,
%2002. 

\bibitem {SU} Scheinerman, E.R. and Ullman, D.H., "Fractional Graph Theory" Wiley,
1997.
\end{thebibliography}
\end{document}